\documentclass[12pt, twoside, letterpaper]{amsart}
\usepackage{amsmath, hyperref, color, amssymb}
\usepackage{srcltx}

\usepackage{geometry}
\newtheorem{thm}{Theorem}[section]
\newtheorem{cor}[thm]{Corollary}
\newtheorem{lem}[thm]{Lemma}

\theoremstyle{definition}

\newtheorem{as}[thm]{Assumption}

\theoremstyle{remark}
\newtheorem{rem}[thm]{Remark}
\newtheorem{exa}[thm]{Example}
\numberwithin{equation}{section}

\newcommand{\R}{\mathbb{R}}

\newcommand{\dbracc}[1]{[\kern-0.15em[ #1 ]\kern-0.15em]}
\newcommand{\dbraco}[1]{[\kern-0.15em[ #1 [\kern-0.15em[}
\newcommand{\dbraoc}[1]{]\kern-0.15em] #1 ]\kern-0.15em]}
\newcommand{\dbraoo}[1]{]\kern-0.15em] #1 [\kern-0.15em[}
\newcommand{\be}{\begin{equation}}
\newcommand{\ee}{\end{equation}}
\newcommand{\ba}{\begin{aligned}}
\newcommand{\ea}{\end{aligned}}

\begin{document}

\title[Sensitivity analysis of the expected utility maximization problem]{Sensitivity analysis of the 
 utility maximization problem with respect to  
 model perturbations
 }

\begin{abstract}

We study the sensitivity of the expected utility maximization problem
in a 
continuous
semi-martingale market with respect  to 
small changes in the market price of risk. Assuming that the preferences of a rational economic agent are modeled with a general utility function, we obtain a second-order expansion of the value function, a first-order approximation of the terminal wealth, and construct trading strategies that match the indirect utility function up to the second order. 
If a risk-tolerance wealth process exists, using it as a num\'eraire  and under an appropriate change of measure, we reduce the approximation problem to a Kunita-Watanabe decomposition.
\end{abstract}


\author{Oleksii Mostovyi}\thanks{We would like to thank Nicolai V. Krylov
for a discussion on the subject of the paper. We would also like to thank
 Kasper Larsen and Gordan \v Zitkovi\'c for their valuable comments.  The first author has been supported by the National Science Foundation under grant No. DMS-1600307 (2015 - 2018), the second author supported by the National Science Foundation under grant No. DMS-1517664 (2015 - 2018). Any opinions, findings, and conclusions or recommendations expressed in this material are those of the authors and do not necessarily reflect the views of the National Science Foundation.}

\address{Oleskii Mostovyi, Department of Mathematics, University of Connecticut, Storrs, CT 06269, United States}%
\email{oleksii.mostovyi@uconn.edu}%

\author{Mihai S\^ irbu}
\address{Mihai S\^ irbu, Department of Mathematics, University of Texas at Austin, Austin, TX 78712, United States}
\email{sirbu@math.utexas.edu}%


\date{\today}%

\maketitle

\section{Introduction}
It is well-known, see for example \cite{DSmathArb, HulleySchweizer}, that for a continuous stock price process, the  no-arbitrage condition implies that the return of the stock price $S$ has the following representation:
\begin{displaymath}
S = M  + \lambda\cdot \langle M\rangle,
\end{displaymath}
where $M$ is a continuous local martingale, and $\lambda$ is a predictable process, i.e., that the quadratic variation of a stock price has to be absolutely continuous with respect to the quadratic variation of $M$.  
 We analyze  {\it the effect of perturbations of the market price of risk $\lambda$}, on the utility maximization problem. 

In the setting of an incomplete model, where the  preferences of a rational economic agent are modeled with a general utility function $U$ with bounded (away from zero and infinity) relative risk-aversion and the stock prices process is continuous, we obtain a quadratic expansion of the value function, a first-order correction to the optimal terminal wealth, and a construction of the approximate trading strategies that match the value functions up to the second order.
 For the power-utility case, a first-order asymptotic expansion with respect to perturbations of the market price of risk is obtained in \cite{RasChau}, whereas a second-order analysis is performed in \cite{LarMosZit14}. Mathematically, the results in the present paper  rely on different techniques. We can summarize our contribution as three-fold:
 \begin{enumerate}

\item We first need to increase dimensionality and look at the simultaneous perturbations of the market price of risk and the initial wealth. As the proofs show, the  increase of dimensionality is a necessary way of getting the expansions of the value functions up to the second order\footnote{In the constant relative risk aversion case considered in \cite{LarMosZit14}, as the optimal terminal wealth depends on the initial wealth via a multiplicative constant, the increase of dimensionality is not needed for obtaining quadratic expansions.}.

 \item Then, we formulate auxiliary quadratic stochastic control problems and relate the second-order approximations of both primal and dual value functions to these problems.   
 
  \item Finally, if the risk-tolerance wealth process exists, we use it as a num\'eraire, and change the measure accordingly, to identify solutions to the general quadratic optimization problems above in terms of a Kunita-Watanabe decomposition (of a certain martingale) generated by the perturbation process. 
  \end{enumerate}


To the best of our knowledge, the closest paper from the mathematical viewpoint is 
\cite{KS2006b}, where the authors obtain a second-order expansion of the value function with respect to simultaneous perturbations of the initial wealth and the number of units of random endowment held in the portfolio. We would like to stress that unlike the present setting, {\it in \cite{KS2006b}, the value function is jointly concave (in both the initial wealth and the number of units of random endowment held in the portfolio)}, a fact that plays a significant role in the proofs there. 

We  combine  here the increase of dimensionality  described in item (1) with  a  similar change of measure and num\' eraire  to \cite{KS2006b} relating them to general quadratic optimization problems. However, one of the main technical difficulties lies in the fact that {\it our value function as a function of two variables is not concave or convex in the perturbation variable $\delta$} (in general). Despite this obstacle, our approach, which relies  only partially on convex conjugacy, still produces a quadratic expansion  via auxiliary quadratic problems and simultaneous expansions of $u$ in $(x,\delta)$ and $v$ in $(y, \delta)$. 
In addition to obtaining a quadratic expansion, we also get a relationship between the existence of such an approximation and the existence of the risk-tolerance wealth process, which was introduced in \cite{KS2006b}. 
We show that the existence of the risk-tolerance wealth process allows for a more explicit form of the correction terms in our approximation coming from a Kunita-Watanabe decomposition under appropriate measure and num\'eraire that are specified in terms of the risk-tolerance wealth process. Another connection to \cite{KS2006b} is given in Lemma \ref{keyConcreteCharLem}, where the perturbation to the market-price of risk plays the role of a multiplicative (and non-linear) random endowment.

To separate the financial aspects of the problem from the mathematical ones, we state and prove {\it abstract versions of main theorems}. After that we reduce the proofs of (some of) the main theorems to verification of the conditions in the abstract theorems.

As an application, we consider models, which admit closed form solutions in incomplete markets, see \cite{KO96, Liu07, GR15} (we also refer to \cite{LarMosZit14} for more examples and a literature review). These models are sensitive to perturbations of the input parameters: ones they are perturbed even slightly, a close form solution typically ceases to exist. Our results show that even though we do not know how to obtain an exact solution for such perturbed problems, an approximation, which is accurate up to the second order, can still be constructed.

We prove our results under the assumption of no unbounded profit with bounded risk, the weakest no-arbitrage type condition, which allows for the utility maximization problem from terminal wealth to be non-degenerate, see \cite[Proposition 4.19]{KarKar07}. For the perturbation process, we formulate an assumption and  give a counterexample, which shows the necessity of the assumption.  In addition, we provide a set of sufficient conditions for the integrability assumption on the perturbation process to hold. 

For the  general utility function, we suppose that its relative risk-aversion is bounded away from zero and infinity. This condition is (essentially) necessary for twice differentiability with respect to the initial wealth to hold, see \cite{KS2006} for counterexamples.
On an even more technical side, as we consider perturbations of the initial wealth, we obtain as a by-product here the second-order derivatives of the primal and dual value functions with respect to the spatial variables ($x$ and $y$, correspondingly). Note that, in \cite{KS2006} this result was obtained for discontinuous stock prices, but under NFLVR. 

The remainder of the paper is organized as follows: in section \ref{secModel}, we formulate the model and state the expansion theorems, section \ref{approxTradingStrategies} contains the approximation of optimal trading strategies theorem, section \ref{secAbstracVersion} includes abstract versions of Theorems \ref{mainThm1}, \ref{mainThm2}, \ref{162}, and \ref{161} with proofs, section \ref{secProofC} contains proofs of non-abstract theorems and Theorem \ref{thmCorOptimizer}, where a construction of corrections to the optimal trading strategies (accurate up to the second order of the value function) is specified. Section \ref{counterexamples} includes a counterexample, which shows that without Assumption \ref{integrabilityAssumption} on the perturbation process, the quadratic expansions of the value functions might not exist. In section \ref{secRiskTol} we relate the asymptotic expansions from previous sections to the existence of the risk-tolerance wealth process and a Kunita-Watanabe decomposition. We conclude the paper with an illustration of an application of our results to analysis of the perturbations of models that admit closed-form solutions.

\section{Model}\label{secModel}
\subsection{Parametrized family of stock prices processes}
Let us consider a complete stochastic basis $\left(\Omega, \mathcal F, \{\mathcal F_t\}_{t\in[0,T]}, \mathbb P\right)$, where $T\in (0,\infty)$ is the time horizon, $\mathcal F$ satisfies the usual conditions, and $\mathcal F_0$ is the completion of the trivial $\sigma$-algebra. We assume that there are two traded securities, a bank account with zero interest rate and a stock%
.
Let $M$ be a one-dimensional continuous local martingale and $\lambda$ is a progressively measurable process, such that
\begin{equation}\label{NUPBR}
\lambda^2\cdot \langle M\rangle_T<\infty,\quad \mathbb {P}-a.s.
\end{equation}
The stock price return process\footnote{We denote the return of the stock by $S$, since $R$ is used for different purposes.} for the unperturbed, or equivalently, $0$-model is given by 
\begin{displaymath}
S^0 \triangleq  \lambda \cdot \langle M\rangle + M.
\end{displaymath}
Here we consider a parametric family of semimartingales $S^\delta$, $\delta\in\mathbb R$, with the same martingale part $M$ and where the market price of risk $\lambda$'s are perturbed
\begin{displaymath}
S^\delta \triangleq  \lambda^\delta \cdot d\langle M\rangle + M,
\end{displaymath}
where for some progressively measurable process $\nu$, such that 
\begin{equation}\label{5182}
\nu^2\cdot \langle M\rangle_T<\infty,\quad \mathbb {P}-a.s.
\end{equation}
we have
\begin{displaymath}
\lambda^\delta \triangleq \lambda + \delta \nu,\quad \delta\in\mathbb R.
 \end{displaymath}
 
\subsection{Primal problem}
Let $U$ be a utility function that satisfies Assumption \ref{rra} below.
\begin{as}\label{rra}
The utility function 
$U$ is strictly increasing, strictly concave, two times differentiable on $(0,\infty)$ and there exist positive constants $c_1$ and $c_2$, such that
\begin{equation}\label{utilityFunction}
{c_1} \leq A(x) \triangleq -\frac{U''(x)x}{U'(x)} \leq c_2,
\end{equation}
i.e. the relative risk aversion of $U$ is uniformly bounded away from zero and infinity.
\end{as}
The family of primal feasible sets is defined as
\begin{equation}\label{primalDomain}
\mathcal X(x,\delta) \triangleq \left\{ 
X\geq 0:~X_t = x + H\cdot S^{\delta}_t,\quad t\in[0, T]
\right\},\quad (x,\delta)\in[0, \infty)\times \mathbb R,
\end{equation}
where $H$ is a predictable and $S^\delta$-integrable process representing the amount invested in the stock.
The corresponding family of the value functions is given by
\begin{equation}\label{primalProblem}
u(x, \delta) \triangleq \sup\limits_{X\in\mathcal X(x,\delta)} \mathbb E\left[ U(X_T)\right],\quad (x,\delta)\in(0, \infty)\times \mathbb R.
\end{equation}
We use the convention
\begin{displaymath}
\mathbb E\left[ U(X_T)\right] \triangleq -\infty, \quad if \quad 
\mathbb E\left[ U^{-}(X_T)\right] = \infty, 
\end{displaymath}
where $U^{-}$ is the negative part of $U$.
%
\subsection{Dual problem}The investigation of the primal problem \eqref{primalProblem} is conducted via the dual problem. 
First, let us define the dual domain as follows:
\begin{equation}\label{dualDomain}
\begin{array}{rccl}\mathcal Y(y, \delta)&\triangleq& \left\{
Y:\right.&Y~is~a~nonnegative~supermartingale,~such~that~Y_0 = y\\
&&& and~ XY = (X_tY_t)_{t\geq 0}~is~a~supermartingale\\
&&& \left.for~every~X\in~\mathcal X(1, \delta)
 \right\},  \quad\quad\quad\quad(y,\delta)\in[0, \infty)\times \mathbb R.
 \end{array}
\end{equation}
We set the convex conjugate to utility function $U$ as
\begin{equation}\label{dualFunction}
V(y) \triangleq \sup\limits_{x > 0} \left( U(x) - xy\right),\quad y>0.
\end{equation}
Note that for $y = U'(x)$, we have 
$$V''(y) = - \frac{1}{U''(x)},$$
and $$B(y) \triangleq -\frac{V''(y)y}{V'(y)}= \frac{1}{A(x)}.$$
Therefore, Assumption \ref{rra} implies that
$$\frac{1}{c_2} \leq B(y)  \leq \frac{1}{c_1}, \quad y> 0.$$
The parametrized family of dual value functions is given by
\begin{equation}\label{dualProblem}
v(y, \delta) \triangleq \inf\limits_{Y\in\mathcal Y(y, \delta)}\mathbb E\left[ 
V(Y_T)
\right], \quad (y,\delta)\in(0, \infty)\times \mathbb R.
\end{equation}
We use the convention
\begin{displaymath}
\mathbb E\left[ 
V(Y_T)
\right] \triangleq \infty, \quad if\quad
\mathbb E\left[ 
V^{+}(Y_T)
\right] = \infty,
\end{displaymath}
where $V^{+}$ is the positive part of $V$.
\section{Technical assumptions}
We recall the assumption that $M$ is continuous.
The absence of arbitrage opportunities in the $0$-model in the sense of no unbounded profit with bounded risk follows from condition \eqref{NUPBR}, which implies that
$
\mathcal Y(1, 0)\neq \emptyset.
$
Note that \eqref{NUPBR} and \eqref{5182} impliy no unbounded profit with bounded risk for every $\delta\in\mathbb R$, thus 
$$\mathcal Y(1, \delta)\neq \emptyset,\quad \delta \in \mathbb R.$$
In order for the problem \eqref{primalProblem} to be non-degenerate, we also need to assume that
\begin{equation}\label{finCond}
u(x,0)<\infty\quad for ~some~x>0.
\end{equation}
\begin{rem}
Conditions \eqref{NUPBR} and \eqref{finCond} are necessary for the expected utility maximization problem to be non-degenerate. Note that we only impose them for $\delta = 0$.
\end{rem}

As in \cite{KS2006, KS2006b}, an important role will be played by the probability measures $\mathbb R(x,\delta)$, given by
\begin{displaymath}
\frac{d \mathbb R(x,\delta)}{d\mathbb P} \triangleq \frac{\widehat X_T(x,\delta) \widehat Y_T(y,\delta)}{xy},
\end{displaymath}
for $x>0$ and $y = u_x(x, \delta)$.
As Example \ref{2142} below demonstrates, we need to impose an integrability condition. 
First, let us define
\begin{equation}\label{defZetaC}
\zeta(c,\delta)\triangleq \exp\left( c(|\nu\cdot S^\delta_T| +  \langle \nu\cdot S^\delta\rangle_T)\right),\quad (c,\delta)\in\mathbb R^2.\end{equation}
\begin{as}\label{integrabilityAssumption}
Let $x>0$ be fixed.
There exists $c>0$, such that 
$$ \mathbb {E}^{\mathbb R(x,0)}\left[\zeta(c,0)\right]<\infty.$$
\end{as}
\begin{rem}The  stronger condition
\begin{equation}\label{2181}
\sup\limits_{(x',\delta)\in B_\varepsilon(x,0)}\mathbb E^{\mathbb R(x',\delta)}\left[\zeta(c,\delta)\right]<\infty,
\end{equation}
for some $\varepsilon>0$ and $c>0$, where $B_\varepsilon(x,0)$ denotes the ball in $\mathbb R^2$ of radius $\varepsilon$ centered at $(x,0)$,
implies {\it local semiconcavity} of the value function $u(x,\delta)$. Consequently, 
in the quadratic expansions of $u$ and $v$ given by \eqref{122121} and \eqref{122122},  the matrices $H_u(x,0)$ and $H_v(y,0)$ defined in \eqref{12136} and \eqref{12137}, respectively, are Hessian matrices, i.e. are derivatives of  gradients.
This will follow from Lemma \ref{12132}. However, the very restrictive condition \eqref{2181} is an assumption that depends on optimal solutions for $\delta \not= 0$, and thus usually impossible to check.
\end{rem}
Let us also set
\begin{equation}\label{Ldeltadef}
L^\delta \triangleq \mathcal E\left(-(\delta\nu)\cdot S^0 \right)_T 
,\quad \delta \in\mathbb R.
\end{equation}
Here and below $\mathcal E$ denotes the Dol\'eans-Dade exponential.
One can see that $L^\delta$ is a terminal value of an element of $\mathcal X(1,0)$ for every $\delta \in\mathbb R$.
\section*{Sufficient conditions for Assumption \ref{integrabilityAssumption}}
\begin{rem}\label{221}A sufficient condition for 
Assumption \ref{integrabilityAssumption} to hold is the existence a wealth process under the num\'eraire $\widehat X(x,0)$, $\widetilde X$, and a constant  $c>0$, such that $$\exp\left( c(|\nu\cdot S^0| + \nu^2\cdot \langle M\rangle)\right)_T\leq \widetilde X_T, \quad a.s.$$ 
\end{rem}

\begin{rem} Let us assume that in \eqref{utilityFunction}, $c_1>1$, i.e. that  relative-risk aversion of $U$ is strictly greater than $1$,  (for example, this holds if $U(x) = \frac{x^p}{p}$ with $p<0$, note that for such a $U$, the conjugate function $V(y) = \frac{y^{-q}}{q}$ for $q\in(-1,0)$). In this case, 
a sufficient condition for Assumption \ref{integrabilityAssumption} to hold
is the existence of some positive exponential moments under $\mathbb P$ of 
$$\left|\nu\cdot S^0_T\right|\quad  and \quad \nu^2\cdot \langle M\rangle_T.$$
This can be shown as follows. Let us set
 $$q_i \triangleq -\left(1-\tfrac{1}{c_i}\right),\quad i = 1,2.$$
 As $c_2\geq c_1>1$, we deduce that $q_i\in(-1,0)$, $i = 1,2.$
Using Lemma \ref{11303}, one can find a constant $C>0$, such that 
\begin{equation}\label{411}
-V'(y)y \leq C\left(y^{-q_1}+y^{-q_2}\right),\quad y>0.
\end{equation}
In order to prove \eqref{411}, let us observe that from Lemma \ref{11303}, we get   
\begin{equation}\label{421}\begin{array}{rcl}
U'(z)&\leq &z^{-c_2}U'(1),\\
-V'(z)&\leq &z^{-\frac{1}{c_1}}(-V'(1)),\quad for~every~z\in(0,1].\\
\end{array}\end{equation}
As $(U')^{-1} = -V'$, the first inequality implies that there exists $z_0$, such that 
$$-V'(z) \leq \left(U'(1)\right)^{\tfrac{1}{c_2}}z^{-\tfrac{1}{c_2}},\quad for~every~z\geq z_0.$$
Combining this inequality with \eqref{421} and since $\sup\limits_{z\in[\min(z_0,1), \max(z_0, 1)]}|-V'(z)z|<\infty,$
 we obtain \eqref{411}. Thus, if some positive exponential moments of 
$\left|\nu\cdot S^0_T\right|$  and $\nu^2\cdot \langle M\rangle_T$
exist under $\mathbb P$, using H\"older's inequality one can find a positive constant $a$, such that 
\begin{equation}\label{12201}
\mathbb E\left[\zeta(a, 0) \right]<\infty,
\end{equation}
where $\zeta(a,0)$ is defined in \eqref{defZetaC}.
Let us set 
$$
c \triangleq a(1+q_2)
$$
 and note that $\frac{c}{1 + q_1} = a\frac{1+q_2}{1+q_1}\leq a.$
With $y = u_x(x,0)$, using H\"older's inequality again (note that $\frac{1}{1+q_i}$ are the H\"older conjugate of $\frac{1}{-q_i}$, $i=1,2$) and \eqref{411}, we get 
\begin{displaymath}
\begin{array}{rcl}
xy\mathbb E^{\mathbb R(x,0)}\left[\zeta(c,0)\right] &\leq& C
 \mathbb E\left[\left(\left(\widehat Y_T(y,0)\right)^{-q_1} + \left(\widehat Y_T(y,0)\right)^{-q_2}\right)\zeta(c,0) \right]\\
 &\leq &
 C\mathbb E\left[ \widehat Y_T(y,0) \right]^{-q_1}\mathbb E\left[\zeta\left(\frac{c}{1+q_1},0\right) \right]^{{1+q_1}} +
  C\mathbb E\left[ \widehat Y_T(y,0) \right]^{-q_2}\mathbb E\left[\zeta\left(\frac{c}{1+q_2},0\right) \right]^{{1+q_2}} \\ 
 &\leq &
  Cy^{-q_1}\mathbb E\left[\zeta({a},0) \right]^{{1+q_1}} +
   Cy^{-q_2}\mathbb E\left[\zeta({a},0) \right]^{{1+q_2}} <\infty,\\
  \end{array}
   \end{displaymath}
   where the last inequality follows from the supermartingale property of $\widehat Y(y,0)$ and \eqref{12201}. Thus, Assumption \ref{integrabilityAssumption} holds.
\end{rem}

\begin{rem}[On the relationship with existing literature]
Assumption \ref{integrabilityAssumption} is related to the condition on random endowment, Assumption 4 in \cite{KS2006b}, via the following argument. 
Assume that, for some $x>0$ and $c>0$, there exists a wealth process $X\in\mathcal X(x,0)$, such that
\begin{equation}\label{5111}
\zeta(c,0)\leq \frac{X_T}{\widehat X_T(x,0)},
\end{equation}
where $\widehat X(x,0)$ is the optimal solution to \eqref{primalProblem}. Then Assumption \ref{integrabilityAssumption} is satisfied. 
 The wealth process $\frac{X}{\widehat X(x,0)}$ under the numeraire $\widehat X(x,0)$ in condition \eqref{5111} is local martingale under $\mathbb R(x,0)$, i.e., $X$ can be an arbitrary element of $\mathcal X(x,0)$. In \cite{KS2006b} it is assumed that $\frac{X}{\widehat X(x,0)}$ is a square-integrable martingale under $\mathbb R(x,0)$.

\end{rem}

\section*{Expansion Theorems}
In 
Theorem \ref{mainThm1} we prove finiteness of the value functions and first-order derivatives with respect to $\delta$. 
\begin{thm}\label{mainThm1}
Let $x>0$ be fixed, assume that \eqref{NUPBR} and \eqref{finCond} as well as Assumptions \ref{rra} and \ref{integrabilityAssumption} hold, and denote $y = u_x(x,0)$, which is well-defined by the abstract theorems in \cite{KS}.  
Then there exists $\delta_0>0$ such that for every $\delta \in(-\delta _0,\delta_0)$, we have
\begin{equation}\label{finitenessConcrete}
u(x,\delta)\in\mathbb R,\quad x>0,\quad and\quad
v(y,\delta)\in\mathbb R,\quad y>0.
\end{equation}
In addition, $u$ and $v$ are jointly differentiable (and, consequently, continuous) at $(x,0)$ and $(y,0)$, respectively. We also have  
\begin{equation}\label{gradientC}
\nabla u (x,0)=\begin{pmatrix}y\\ u_{\delta}(x,0)\end{pmatrix}\quad and \quad \nabla v (y,0)=\begin{pmatrix}-x\\v_{\delta}(y,0)\end{pmatrix},
\end{equation} where
\begin{equation}\label{firstOrderDerivConcrete}
u_\delta(x,0) =v_\delta(y,0)= xy\mathbb E^{\mathbb R(x,0)}\left[ F\right].
\end{equation}

\end{thm}
In order to characterize the second-order derivatives of the value functions, we will need the following notations. 
Let $S^{X(x, 0)}$ be the price process of the traded securities under the num\'eraire $\frac{\widehat X(x,0)}{x}$, i.e.
\begin{displaymath}
S^{X(x, 0)} = \left( \frac{x}{\widehat X(x,0)}, \frac{xS^0}{\widehat X(x,0)}\right).
\end{displaymath}
For every $x>0$, let $\mathbf H^2_0(\mathbb R(x, 0))$ denote the space of square integrable martingales under $\mathbb R(x, 0)$, such that
\begin{displaymath}
\begin{array}{rcl}
\mathcal M^2(x, 0) &\triangleq &\left\{M\in\mathbf H_0^2(\mathbb R(x, 0)):  M = H\cdot S^{X(x, 0)} \right\},\\
\mathcal N^2(y, 0) &\triangleq &\left\{N\in\mathbf H_0^2(\mathbb R(x, 0)):  MN~is~\mathbb R(x, 0)-martingale~for~every~M\in\mathcal M^2(x, 0)\right\},\\
&&\hspace{106mm} here~y= u_x(x,0).\\
\end{array}
\end{displaymath}

\section*{Auxiliary minimization problems}
As in \cite{KS2006}, for $x>0$ let us consider 
\begin{equation}\label{axxC}
a(x,x) \triangleq \inf\limits_{M \in\mathcal M^2(x, 0)}\mathbb E^{\mathbb R(x, 0)}\left[ A(\widehat X_T(x,0))(1 + M_T)^2\right],
\end{equation}
\begin{equation}\label{byyC}
b(y,y) \triangleq \inf\limits_{N \in\mathcal N^2(y, 0)}\mathbb E^{\mathbb R(x, 0)}\left[ B(\widehat Y_T(y,0))(1 + N_T)^2\right],\quad y = u_x(x,0),
\end{equation}
where $A$ is the relative risk aversion and $B$ is the relative risk tolerance of $U$, respectively. It is proven in \cite{KS2006}\footnote{Under the assumption of NFLVR. Below we will show that the formulas \eqref{axxC}, \eqref{byyC}, and \eqref{122111C} can also be obtained in the present setting.} that \eqref{axxC} and \eqref{byyC} admit unique solutions $M^0(x,0)$ and $N^0(y,0)$, correspondingly, and 
\begin{equation}\label{122111C}
\begin{array}{rcl}
u_{xx}(x,0)&=&-\frac{y}{x}a(x,x),\\
v_{yy}(y,0)&=&\frac{x}{y}b(y,y),\\ 
a(x,x)b(y,y) &=& 1,\\
A(\widehat X_T(x,0))(1 + M^0_T(x,0)) &=& a(x,x)(1+ N^0_T(y,0)).\\
\end{array}
\end{equation}
In order to characterize the derivatives of the value functions with respect to $\delta$, with
\begin{equation}\label{defFG}
F\triangleq \nu\cdot S^0_T\quad and \quad G\triangleq\nu^2\cdot \langle M\rangle_T,
\end{equation}
we consider the following minimization problems:
\begin{equation}\label{addC}
\begin{array}{rcl}
a(d,d) &\triangleq& \inf\limits_{M\in\mathcal M^2(x, 0)}\mathbb E^{\mathbb R(x, 0)}\left[A(\widehat X_T(x,0))(M_T + xF)^2 - 
2xFM_T - x^2(F^2 + G) \right],\\ 
\end{array}
\end{equation}
\begin{equation}\label{bddC}
\begin{array}{rcl}
b(d,d) &\triangleq& \inf\limits_{N\in\mathcal N^2(y, 0)}\mathbb E^{\mathbb R(x, 0)}\left[B(\widehat Y_T(y,0))(N_T - yF)^2 + 
2yFN_T - y^2(F^2 - G) \right],\\ 
\end{array}
\end{equation}
Denoting by $M^1(x,0)$ and $N^1(y,0)$ the unique solutions to \eqref{addC} and \eqref{bddC} respectively, we also set
\begin{equation}\label{axdC}
a(x,d) \triangleq \mathbb E^{\mathbb R(x, 0)}\left[A(\widehat X_T(x,0))(1 + M^0_T(x,0))(xF + M^1_T(x,0)) - xF(1 + M^0_T(x,0)) \right],
\end{equation}
\begin{equation}\label{bydC}
b(y,d) \triangleq \mathbb E^{\mathbb R(x, 0)}\left[B(\widehat Y_T(y,0))(1 + N^0_T(y,0))(N^1_T(y,0)-yF) + yF(1 +N^0_T(y,0)) \right].
\end{equation}
Theorems  \ref{mainThm2}, \ref{162}, and \ref{161} contain the second-order expansions of the value functions, derivatives of the optimizers, and properties of such derivatives.
\begin{thm}\label{mainThm2}
Let $x>0$ be fixed.  Assume all  conditions of Theorem \ref{mainThm1} hold,  with $y = u_x(x,0)$. 
 Define 
\begin{equation}\label{12136C}
H_u(x,0) \triangleq  -\frac{y}{x}\begin{pmatrix} 
      a(x,x) 		& a(x, d)\\ 
  a(x,d)	& a(d, d)\\ 
\end{pmatrix},
\end{equation}
where $a(x,x)$, $a(d,d)$, and $a(x, d)$ are specified in \eqref{axxC}, \eqref{addC}, and \eqref{axdC},  and, respectively,
\begin{equation}\label{12137}
H_v(y,0) \triangleq \frac{x}{y}\begin{pmatrix} 
      b(y,y) 		& b(y, d)\\ 
  b(y,d)	& b(d, d)\\ 
\end{pmatrix},
\end{equation}
where  $b(y,y)$, $b(d,d)$, $b(y, d)$ are specified in \eqref{byyC}, \eqref{bddC}, and \eqref{bydC}.
Then, the value functions $u$ and $v$ admit the  second-order expansions  around $(x,0)$ and $(y,0)$, respectively, 
\begin{equation}\label{122121}
u(x+\Delta x,\delta) = u(x,0) + (\Delta x\quad \delta) \nabla u(x,0) + \tfrac{1}{2}(\Delta x\quad \delta) 
H_u(x,0)
\begin{pmatrix}
\Delta x\\
\delta\\
\end{pmatrix} + o(\Delta x^2 + \delta^2),
\end{equation}
and
\begin{equation}\label{122122}
v(y+\Delta y,\delta) = v(y,0) + (\Delta y\quad \delta) \nabla v(y,0) + \tfrac{1}{2}(\Delta y\quad \delta) 
H_v(y,0)
\begin{pmatrix}
\Delta y\\
\delta\\
\end{pmatrix} + o(\Delta y^2 + \delta^2).
\end{equation}\end{thm}
\begin{rem}\label{rem:415}Although we only have second order expansions, we may abuse the language and call  $H_u(x,0)$ and $H_v(y,0)$ the Hessians of $u$ and $v$, without having twice differentiability. This   causes no confusion, see the discussion e.g., in \cite{Lewis02}. The meaning of partial derivatives $u_{xx} (x,0), u_{x\delta}(x,0)$ and so on then becomes apparent by identifying entries in the Hessian matrices.
\end{rem}
\begin{thm}\label{162}
Let $x>0$ be fixed, the assumptions of Theorem \ref{mainThm1} hold, and $y = u_x(x,0)$. Then, we have
\begin{equation}\label{12154c}
\begin{pmatrix}
a(x,x) & 0\\
a(x,d) &-\frac{x}{y} \\
\end{pmatrix}
\begin{pmatrix}
b(y,y) & 0\\
b(y,d) & -\frac{y}{x} \\
\end{pmatrix} = I_2,
\end{equation}
where $I_2$ denotes two-by-two identity matrix. 
Moreover,
\begin{equation}\label{keyGapc}
\frac{y}{x}a(d,d) + \frac{x}{y}b(d,d) = a(x,d)b(y,d),
\end{equation}
\begin{equation}\label{4141}
U''(\widehat X_T(x,0))\widehat X^0_T(x,0)\begin{pmatrix}
M^0_T(x,0) + 1\\
M^1_T(x,0) + xF
\end{pmatrix} = 
-\begin{pmatrix}
a(x,x) & 0 \\
a(x,d) & -\frac{x}{y} 
\end{pmatrix}
\widehat Y^0_T(y,0)
\begin{pmatrix}
N^0_T(y,0) + 1\\
N^1_T(y,0) - yF
\end{pmatrix},
\end{equation}
\begin{equation}\nonumber
V''(\widehat Y_T(y,0))\widehat Y_T(y,0)\begin{pmatrix} 1 + N^0_T(y,0) \\ -yF + N^1_T(y,0)\\\end{pmatrix} = 
\begin{pmatrix}
b(y,y)& 0\\b(y,d) &- \frac{y}{x}\end{pmatrix}\widehat X_T(x,0)
\begin{pmatrix}
1 + M^0_T(x,0) \\ xF + M^1_T(x,0)\\
\end{pmatrix}.
\end{equation}
and the product of any of $\widehat X(x,0)$, $\widehat X(x,0)M^0(x,0)$, $\widehat X(x,0)M^1(x,0)$ and any of $\widehat Y(y,0)$, $\widehat Y(y,0)N^0(y,0)$, $\widehat Y(y,0)N^1(y,0)$ is a martingale under $\mathbb P$, where $M^0_T(x,0)$, $M^0_T(x,0)$, $N^0_T(y,0)$, and $N^1_T(y,0)$ are the solutions to 
\eqref{axxC}, \eqref{addC}, \eqref{byyC}, and \eqref{bddC}, correspondingly.

\end{thm}

\begin{rem}
Continuing the discussion in Remark \ref{rem:415},
\eqref{12154c} implies that
\begin{displaymath}
\begin{pmatrix}
u_{xx}(x,0) & 0 \\
u_{x\delta}(x,0) & 1 \\
\end{pmatrix}\begin{pmatrix}
v_{yy}(y,0) & 0\\
v_{y\delta}(y,0) & -1\\
\end{pmatrix}= -I_2,
\end{displaymath}
where
\begin{displaymath}
 \begin{pmatrix}
u_{xx}(x,0) & 0 \\
u_{x\delta}(x,0) & 1 \\
\end{pmatrix} = -\frac{y}{x}\begin{pmatrix}
a(x,x) & 0\\
a(x,d) &-\frac{x}{y} \\
\end{pmatrix}\quad and
\quad
\begin{pmatrix}
v_{yy}(y,0) & 0\\
v_{y\delta}(y,0) & -1\\
\end{pmatrix} = \frac{x}{y}\begin{pmatrix}
b(y,y) & 0\\
b(y,d) & -\frac{y}{x} \\
\end{pmatrix}.
\end{displaymath}
Likewise, \eqref{keyGapc} gives
\begin{equation}\nonumber
-u_{\delta\delta}(x,0) + v_{\delta\delta}(y,0) = -u_{x\delta}(x,0)v_{y\delta}(y,0).
\end{equation}
\end{rem}
\begin{thm}\label{161}
Let $x>0$ be fixed, the assumptions of Theorem \ref{mainThm1} hold, and $y = u_x(x,0)$. Then the terminal values of the wealth processes $M^0(x,0)$ and $M^1(x,0)$, which are the solutions to \eqref{axxC} and \eqref{addC}, respectively,  satisfy
\begin{equation}\label{1251c}
\begin{array}{rcl}
\lim\limits_{|\Delta x| + |\delta| \to 0}\frac{1}{|\Delta x| + |\delta|}\left|
\widehat X_T(x + \Delta x, \delta) - \frac{\widehat X_T(x,0)}{x}\left(x + \Delta x (1 + M^0_T(x,0)) + \delta M^1_T(x,0) \right)\frac{1}{L^{\delta}} \right|&=& 0,\\
\end{array}
\end{equation}
where the convergence takes place in $\mathbb P$-probability and $L^\delta$'s are defined in \eqref{Ldeltadef}. Likewise, let $N^0_T(y,0)$ and $N^1_T(y,0)$, which are solutions to \eqref{byyC} and \eqref{bddC}, correspondingly, satisfy
\begin{equation}\label{1252}
\begin{array}{rcl}
\lim\limits_{|\Delta y| + |\delta |\to 0}\frac{1}{|\Delta y| + |\delta|}\left|
\widehat Y_T(y + \Delta y, \delta) - \frac{\widehat Y_T(y,0)}{y}\left(y + \Delta y(1 + N^0_T(y,0) )+ \delta N^1_T(y,0) \right)L^\delta \right|&=& 0,\\
\end{array}
\end{equation}
where the convergence takes place in $\mathbb P$-probability.
\end{thm}
One can obtain the following corollary.
\begin{cor}\label{2151}
Let $x>0$ be fixed, the assumptions of Theorem \ref{mainThm1} hold,  and $y = u_x(x,0)$. Then, if we define
$$X'_T(x, 0)\triangleq \frac{\widehat X_T(x,0)}{x} (1 + M^0_T(x,0)),\ \ Y'_T(y,0)\triangleq 
\frac{\widehat Y_T(y,0)}{y}(1 + N^0_T(y,0) ),$$ 
and
$$X^d_T(x,0)\triangleq  \frac{\widehat X_T(x,0)}{x} (M^1_T(x,0) +xF),\ \ Y^d_T(y,0)\triangleq 
\frac{\widehat Y_T(y,0)}{y}(N^1_T(y,0) - yF),$$
we have
\begin{equation}\nonumber
\begin{array}{rcl}
\lim\limits_{|\Delta x| + |\delta| \to 0}\frac{1}{|\Delta x| + |\delta|}\left|
\widehat X_T(x + \Delta x, \delta) -\widehat X_T(x,0)- \Delta x X'_T(x,0) -\delta X^d_T(x,0)\right|&=& 0,\\
\end{array}
\end{equation}
\begin{equation}\nonumber
\begin{array}{rcl}
\lim\limits_{|\Delta y| + |\delta |\to 0}\frac{1}{|\Delta y| + |\delta|}\left|
\widehat Y_T(y + \Delta y, \delta) - \widehat Y_T(y,0)-\Delta y Y'_T(y,0)-\delta Y^d_T(y,0)\right|&=& 0,\\
\end{array}
\end{equation}
where the convergence takes place in $\mathbb P$-probability.
\end{cor}
\begin{rem}
Even though Corollary \ref{2151} gives a more explicit form of the derivatives of the terminal wealth, an approximation given in \eqref{1251c} turns out to be more useful in applications.
\end{rem}

\section{Approximation of the optimal trading strategies}
\label{approxTradingStrategies}
Below in this section we will suppose that $x>0$ is fixed.
Let us denote \begin{equation}\label{4231}
M^R \triangleq S^0 - \widehat \pi(x,0)\cdot \langle M\rangle,
\end{equation}
where $\widehat \pi(x,0) = (\widehat \pi_t(x,0))_{t\in[0,T]}$ is the optimal {\it proportion} 
invested in stock corresponding the initial wealth $x$ and $\delta =0$.
Note that for every predictable pair of processes $G^1$ and $G^2$, such that both integrals $G^1\cdot \left(\frac{1}{\widehat X(1,0)}\right)$ and $G^2\cdot \left(\frac{S^0}{\widehat X(1,0)}\right)$ are well-defined, by direct computations, we can find a process $G$, such that $$G^1\cdot \left(\frac{1}{\widehat X(1,0)}\right) + G^2\cdot\left( \frac{S^0}{\widehat X(1,0)} \right)= G\cdot M^R.$$ 

Let $\gamma^0$ and $\gamma^1$ be such that $$\gamma^0\cdot M^R = \frac{M^0(x,0)}{x}\quad and \quad \gamma^1\cdot M^R = \frac{M^1(x,0)}{x}.$$
We need to define the following families of stopping times.
\begin{equation}\nonumber
\begin{array}{rcl}
\sigma_\varepsilon &\triangleq &\inf\left\{ t\in [0,T]:~|M^0_t(x,0)|\geq \frac{x}{\varepsilon}~or~\langle M^0(x,0)\rangle_t\geq \frac{x}{\varepsilon}\right\}, \\
\tau_\varepsilon &\triangleq &\inf\left\{ t\in [0,T]:~|M^1_t(x,0)|\geq \frac{x}{\varepsilon}~or~\langle M^1(x,0)\rangle_t\geq \frac{x}{\varepsilon}\right\},\quad\varepsilon>0, \\
\end{array}
\end{equation} 
we also set
\begin{equation}\nonumber
\gamma^{0,\varepsilon} = \gamma^01_{\{[0,\sigma_\varepsilon]\}}\quad and \quad \gamma^{1,\varepsilon} = \gamma^11_{\{[0,\tau_\varepsilon]\}},\quad \varepsilon>0.
\end{equation}
Finally, for every $(\Delta x, \delta, \varepsilon)\in(-x,\infty)\times\mathbb R\times (0,\infty)$, let us define
\begin{equation}\label{4232}
 X^{\Delta x, \delta, \varepsilon} \triangleq (x + \Delta x)\mathcal E\left(\left(\widehat \pi(x,0) + \Delta x\gamma^{0,\varepsilon} + \delta(\nu+ \gamma^{1,\varepsilon})\right)\cdot S^\delta\right).
\end{equation} 

\begin{thm}
\label{thmCorOptimizer}
Assume that $x>0$ is fixed and the assumptions of Theorem \ref{mainThm1} hold. 
Then, there exists a function $\varepsilon = \varepsilon(\Delta x, \delta)$, $(\Delta x,\delta)\in(-x,\infty)\times\mathbb R$, such that
\begin{equation}\nonumber
\mathbb E\left[ U\left(X^{\Delta x, \delta, \varepsilon(\Delta x, \delta)}_T\right) \right] = u(x + \Delta x, \delta) - o(\Delta x^2 + \delta^2),
\end{equation}
where $X^{\Delta x, \delta, \varepsilon}$ is defined in \eqref{4232}.
\end{thm}
\begin{rem}
Theorem \ref{thmCorOptimizer} shows how to correct the optimal proportion in order to match the primal value function up to the second order {\it jointly in $(\Delta x,\delta)$}.
\end{rem}
\begin{rem}
Proportions have a nicer representation of the corrections to optimal trading strategies in terms of the quadratic optimization problems \eqref{axxC}  and \eqref{addC} because the optimal wealth process was used as num\'eraire, i.e., $\hat{X}(x,0)/x$  has a multiplicative structure.  The result in Theorem \ref{thmCorOptimizer} compliments the results in \cite{KS2006}  and (in a  different additive random endowment framework) those in \cite{KS2006b} in the context of a one-dimensional and continuous stock model.
\end{rem}
\section{Abstract version}\label{secAbstracVersion}

\section*{Abstract version for $0$-model}

We begin with the formulation of the abstract version for $0$-model. As in \cite{Mostovyi2015}, let $(\Omega, \mathcal F, \mathbb P)$ be a measure space and we define the sets $\mathcal C$ and $\mathcal D$ to be subsets of $\mathbf L^0_{+}$ that satisfy the following assumption. Note that Assumption \ref{bipolar} is the abstract version of no unbounded profit with bounded risk condition \eqref{NUPBR}.
\begin{as}\label{bipolar}
Both $\mathcal C$ and $\mathcal D$ contain a stricly positive element and
$$\xi\in\mathcal C\quad iff\quad \mathbb E\left[\xi\eta\right]\leq 1\quad for~every~\eta\in\mathcal D,$$
as well as
$$\eta\in\mathcal D\quad iff\quad \mathbb E\left[\xi\eta\right]\leq 1\quad for~every~\xi\in\mathcal C.$$
\end{as}
We also set $\mathcal C(x, 0) \triangleq x\mathcal C$ and $\mathcal D(x, 0) \triangleq x\mathcal D$, $x>0$.
Now we can state the abstract primal and dual problems as
\begin{equation}\label{abstractPrimal}
u(x,0) \triangleq \sup\limits_{\xi\in\mathcal C(x, 0)}\mathbb E\left[ U(\xi)\right],\quad x>0,
\end{equation}
\begin{equation}\label{abstractDual}
v(y,0) \triangleq \inf\limits_{\eta\in\mathcal D(y, 0)}\mathbb E\left[ V(\eta)\right],\quad y>0.
\end{equation}
Under finiteness of both primal and dual value functions on $\mathbb R$, existence and uniqueness of solutions to \eqref{abstractPrimal} and \eqref{abstractDual} follow from \cite[Theorem 3.2]{Mostovyi2015}. Likewise, with a deterministic utility function that has reasonable asymptotic elasticity, if $u(x, 0)<\infty$ for some $x>0$, standard conclusions of the utility maximization theory also follow from the abstract theorems in \cite{KS} (see the discussion in \cite[Remark 2.5]{MostovyiNUPBR}).
\section*{Abstract version for $\delta$-models}
For some random variables $F$ and $G\geq 0$, let us set
\begin{equation}\label{Ldeltadef}
L^\delta \triangleq \exp\left( -(\delta F + \tfrac{1}{2}\delta^2 G)\right),
\end{equation}
\begin{equation}\label{bipolarL}
\mathcal C(x,\delta) \triangleq \mathcal C(x,0)\frac{1}{L^\delta}\quad and \quad\mathcal D(y,\delta) \triangleq \mathcal D(y,0){L^\delta},\quad \delta \in\mathbb R.
\end{equation}
Now, we can state the abstract versions of the perturbed optimization problems.
\begin{equation}\label{perturbedAbstractPrimal}
u(x,\delta) \triangleq \sup\limits_{\xi\in\mathcal C(x, \delta)}\mathbb E\left[ U(\xi)\right] = 
\sup\limits_{\xi\in\mathcal C(x,0)}\mathbb E\left[ U\left(\xi \frac{1}{L^\delta}\right)\right] ,\quad (x,\delta) \in(0,\infty)\times\mathbb R,
\end{equation}
\begin{equation}\label{perturbedAbstractDual}
v(y,\delta) \triangleq \inf\limits_{\eta\in\mathcal D(y, \delta)}\mathbb E\left[V(\eta)\right] = 
\inf\limits_{\eta\in\mathcal D(y,0)}\mathbb E\left[ V\left(\eta {L^\delta}\right)\right] ,\quad (y,\delta) \in(0,\infty)\times\mathbb R.
\end{equation}
Under an appropriate integrability assumption specified below, existence and uniqueness of solutions to \eqref{perturbedAbstractPrimal} and \eqref{perturbedAbstractDual} as well as conjugacy relations between $u(\cdot, \delta)$ and $v(\cdot, \delta)$ for every $\delta$ sufficiently close to $0$ will follow from \cite[Theorem 3.2]{Mostovyi2015}.

\section*{Condition on perturbations}

Let $\xi(x,\delta)$ and $\eta(y,\delta)$ denote the solutions to \eqref{perturbedAbstractPrimal} and \eqref{perturbedAbstractDual}, respectively, if such solutions exist. By $\mathbb R(x,\delta)$ we denote the probability measure on $(\Omega, \mathcal F)$, whose Radon-Nikodym derivative with respect to $\mathbb P$ is given by
\begin{equation}\label{abstractR}
\frac{d\mathbb R(x,\delta)}{d\mathbb P}\triangleq \frac{\xi(x,\delta)\eta(y,\delta)}{xy},
\end{equation}
where $x>0$, $\delta \in\mathbb R$, and $y = u_x(x,\delta)$.
\begin{as}\label{integrabilityAs}
Let there exists $c>0$, such that 
$$ \mathbb {E}^{\mathbb R(x,0)}\left[ \exp\left( c(|F| + G)\right)\right]<\infty.$$
\end{as}
Note that, $\mathbb R(x,0)$ is well-defined for every $x>0$.

\section*{Expansion theorems}
\section*{Auxiliary sets $\mathcal A$ and $\mathcal B$}
As in \cite{KS2006}, for every $x>0$ and $\delta \in\mathbb R$, we denote by $\mathcal A^\infty(x, \delta)$ the family of bounded random variables $\alpha$, such that $\xi(x,\delta)(1+c\alpha)$ and $\xi(x,\delta)(1 - c\alpha)$ belong to $\mathcal C(x,\delta)$ for some constant $c = c(\alpha)>0$, that is

\begin{equation}\label{Ainfty}
\mathcal A^\infty(x,\delta)\triangleq\left\{\alpha\in\mathbf L^\infty:\xi(x,\delta)(1\pm c\alpha)\in\mathcal C(x,\delta)~for~some~c>0\right\}.
\end{equation}Likewise, for $y>0$ and $\delta \in\mathbb R$, we set
\begin{equation}\label{Binfty}
\mathcal B^\infty(y,\delta)\triangleq\left\{\beta\in\mathbf L^\infty:\eta(y,\delta)(1\pm c\beta)\in\mathcal D(y,\delta)~for~some~c>0\right\}.
\end{equation}
It follows from the Assumption \ref{bipolar} that for every $x>0$, $\mathcal A^\infty(x,\delta)$ and $\mathcal B^\infty(u_x(x,\delta), \delta)$ are orthogonal linear subspaces of $\mathbf L^2_0(\mathbb R(x,\delta)).$

Let us denote by $\mathcal A^2(x,\delta)$ and $\mathcal B^2(y,\delta)$ the respective closures of $\mathcal A^\infty(x,\delta)$ and $\mathcal B^\infty(y,\delta)$ in $\mathbf L^2_0(\mathbb R(x,\delta))$. One can see that $\mathcal A^2(x,\delta)$ and $\mathcal B^2(y,\delta)$ are closed orthogonal linear subspaces of $\mathbf L^2(\mathbb R(x,\delta))$. In order to make these sets related to the concrete versions of the expansion theorems, we need the following assumption.
\begin{as}\label{complimentarity}
For every $\delta \in \mathbb R$ and $x>0$, with $y = u_x(x,\delta)$, the sets $\mathcal A^2(x,\delta)$ and $\mathcal B^2(y,\delta)$ are complimentary linear subspaces in $\mathbf L^2(\mathbb R(x,\delta))$, i.e.
\begin{equation}\label{complimentarityEq}
\begin{array}{rclclc}
\alpha \in \mathcal A^2(x,\delta) &iff& \alpha\in\mathbf L^2_0(\mathbb R(x,\delta)) &and& \mathbb E^{\mathbb R(x,0)}\left[\alpha\beta\right] = 0,& for~every~\beta\in\mathcal B^2(y,\delta),\\
\beta \in \mathcal B^2(y,\delta) &iff& \beta\in\mathbf L^2_0(\mathbb R(x,\delta)) &and& \mathbb E^{\mathbb R(x,0)}\left[\alpha\beta\right] = 0,& for~every~\alpha\in\mathcal A^2(x,\delta).\\
\end{array}
\end{equation}
\end{as}

The following theorem shows joint continuity, and differentiability, and is a consequence of the second-order expansion.
\begin{thm}\label{byproduct}
Let $x>0$ be fixed. Suppose that assumptions \ref{rra}, \ref{bipolar}, \ref{integrabilityAs}, and \ref{complimentarity} hold, $u(z,0)<\infty$ for some $z>0$, and $y=u_x(x,0)$, which is well-defined by the abstract theorems in \cite{KS}.  Then there exists $\delta_0>0$ such that for every $\delta \in(-\delta _0,\delta_0)$, we have
\begin{equation}\label{1244}
u(x,\delta)\in\mathbb R,\quad x>0,\quad and\quad
v(y,\delta)\in\mathbb R,\quad y>0.
\end{equation}
In addition, $u$ and $v$ are jointly differentiable (and, consequently, continuous) at $(x,0)$ and $(y,0)$, respectively. We also have
\begin{equation}\label{gradient}\nabla u (x,0)=\begin{pmatrix}y\\ u_{\delta}(x,0)\end{pmatrix}\quad and \quad \nabla v (y,0)=\begin{pmatrix}-x\\v_{\delta}(y,0)\end{pmatrix},
\end{equation} where
\begin{displaymath}
u_\delta(x,0) =v_\delta(y,0)= xy\mathbb E^{\mathbb R(x,0)}\left[ F\right].
\end{displaymath}
\end{thm}

\begin{rem}
It is possible to prove Theorem \ref{byproduct} without Assumption \ref{complimentarity}. We do not present such a proof for brevity of exposition.
\end{rem}

\section*{Auxiliary minimization problems}
As in \cite{KS2006}, for $x>0$, let us consider 
\begin{equation}\label{axx}
a(x,x) \triangleq \inf\limits_{\alpha \in\mathcal A^2(x,0)}\mathbb E^{\mathbb R(x,0)}\left[ A(\xi(x,0))(1 + \alpha)^2\right],
\end{equation}
\begin{equation}\label{byy}
b(y,y) \triangleq \inf\limits_{\beta \in\mathcal B^2(y,0)}\mathbb E^{\mathbb R(x,0)}\left[ B(\eta(y,0))(1 + \beta)^2\right],\quad y = u_x(x,0),
\end{equation}
where $A$ is the relative risk aversion and $B$ is the relative risk tolerance of $U$, respectively. It is proven in \cite{KS2006} that 
\begin{equation}\label{122111}
\begin{array}{rcl}
u_{xx}(x,0)&=&-\frac{y}{x}a(x,x),\\
v_{yy}(y,0)&=&\frac{x}{y}b(y,y),\\ 
a(x,x)b(y,y) &=& 1,\\
A(\eta(x,0))(1 + \alpha(x,0)) &=& a(x,x)(1+ \beta(y,0)),\\
\end{array}
\end{equation}
where $\alpha(x,0)$ and $\beta(y,0)$ are the unique solutions to \eqref{axx} and \eqref{byy} respectively. In order to characterize derivatives of the value functions with respect to $\delta$, we consider the following minimization problems:
\begin{equation}\label{add}
\begin{array}{rcl}
a(d,d) &\triangleq& \inf\limits_{\alpha\in\mathcal A^2(x,0)}\mathbb E^{\mathbb R(x,0)}\left[A(\xi(x,0))(\alpha + xF)^2 - 
2xF\alpha - x^2(F^2 + G) \right],\\ 
\end{array}
\end{equation}
\begin{equation}\label{bdd}
\begin{array}{rcl}
b(d,d) &\triangleq& \inf\limits_{\beta\in\mathcal B^2(y,0)}\mathbb E^{\mathbb R(x,0)}\left[B(\eta(y,0))(\beta - yF)^2 + 
2yF\beta - y^2(F^2 - G) \right],\\ 
\end{array}
\end{equation}
Denoting by $\alpha_{d}(x,0)$ and $\beta_{d}(y,0)$ the unique solutions to \eqref{add} and \eqref{bdd} respectively, we also set
\begin{equation}\label{axd}
a(x,d) \triangleq \mathbb E^{\mathbb R(x,0)}\left[A(\xi(x,0))(1 + \alpha(x,0))(xF + \alpha_{d}(x,0)) - xF(1 + \alpha(x,0)) \right],
\end{equation}
\begin{equation}\label{byd}
b(y,d) \triangleq \mathbb E^{\mathbb R(x,0)}\left[B(\eta(y,0))(1 + \beta(y,0))(-yF + \beta_{d}(y,0)) + yF(1 + \beta(y,0)) \right].
\end{equation}
We are ready to state the following theorem.
\begin{thm}\label{seconOrderThmA} Let $x>0$ be fixed, the conditions of Theorem \ref{byproduct} hold, and $y = u_x(x,0)$. 
Define 
\begin{equation}\label{12136}
H_u(x,0)\triangleq  -\frac{y}{x}\begin{pmatrix} 
      a(x,x) 		& a(x, d)\\ 
  a(x,d)	& a(d,d)\\ 
\end{pmatrix},
\end{equation}
where $a(x,x)$, $a(d,d)$, and $a(x, d)$ are specified in \eqref{axx}, \eqref{add}, and \eqref{axd}, respectively; and 
\begin{equation}\label{12137}
H_v(y,0) \triangleq  \frac{x}{y}\begin{pmatrix} 
      b(y,y) 		& b(y, d)\\ 
  b(y,d)	& b(d,d)\\ 
\end{pmatrix},
\end{equation}
and in turn, $b(y,y)$, $b(d,d)$, $b(y,d)$ are specified in \eqref{byy}, \eqref{bdd}, and \eqref{byd}, correspondingly. 
Using the formula for the gradients \eqref{gradient}, 
the second-order expansions of the value functions are given by 
\begin{equation}\label{122121}
u(x+\Delta x,\delta) = u(x,0) + (\Delta x\quad \delta) \nabla u(x,0) + \tfrac{1}{2}(\Delta x\quad \delta) 
H_u(x,0)
\begin{pmatrix}
\Delta x\\
\delta\\
\end{pmatrix} + o(\Delta x^2 + \delta^2),
\end{equation}
and
\begin{equation}\label{122122}
v(y+\Delta y,\delta) = v(y,0) + (\Delta y\quad \delta) \nabla v(y,0) + \tfrac{1}{2}(\Delta y\quad \delta) 
H_v(y,0)
\begin{pmatrix}
\Delta y\\
\delta\\
\end{pmatrix} + o(\Delta y^2 + \delta^2).
\end{equation}
\end{thm}
\section*{Derivatives of the optimizers}
\begin{thm}\label{closingGapThm}
Let $x>0$ be fixed, the conditions of Theorem \ref{byproduct} hold, and $y = u_x(x,0)$. Let \textcolor{black}{$\xi =  \xi(x,0)$ and $\eta =  \eta(y,0)$ denote the solutions to \eqref{abstractPrimal} and \eqref{abstractDual}, $\alpha = \widehat \alpha(x,0)$, $\beta = \widehat \beta(y,0)$, $\alpha_{d} = \widehat \alpha_{d}(x,0)$, and $\beta_{d} = \widehat \beta_{d}(y,0)$ denote the solutions to \eqref{axx}, \eqref{byy}, \eqref{axd}, and \eqref{byd}, respectively.} Then, we have
\begin{equation}\label{12154}
\begin{pmatrix}
a(x,x) & 0\\
a(x,d) &-\frac{x}{y} \\
\end{pmatrix}
\begin{pmatrix}
b(y,y) & 0\\
b(y,d) & -\frac{y}{x} \\
\end{pmatrix} = I_2.
\end{equation}
Moreover,
\begin{equation}\label{keyGap}
\frac{y}{x}a(d,d) + \frac{x}{y}b(d,d) = a(x,d)b(y,d)
\end{equation}
and
\begin{equation}\label{keyGap2}
A(\xi)\begin{pmatrix} 1 + \alpha \\ xF + \alpha_{d}\\\end{pmatrix} = 
\begin{pmatrix}
a(x,x)& 0\\a(x,d) &- \frac{x}{y}\end{pmatrix}
\begin{pmatrix}
1 + \beta \\ -yF + \beta_{d}\\
\end{pmatrix},
\end{equation}
equivalently
\begin{equation}\label{12251}
B(\eta)\begin{pmatrix} 1 + \beta \\ -yF + \beta_{d}\\\end{pmatrix} = 
\begin{pmatrix}
b(y,y)& 0\\b(y,d) &- \frac{y}{x}\end{pmatrix}
\begin{pmatrix}
1 + \alpha \\ xF + \alpha_{d}\\
\end{pmatrix}.
\end{equation}
\end{thm}
\begin{thm}\label{134}
Let the conditions of Theorem \ref{byproduct} hold and  $x>0$ be fixed. Then the random variables $\alpha$ and $\alpha_{d}$, which are solutions to \eqref{axx} and \eqref{add}, respectively,  are the partial derivatives of the solution $\widehat \xi(x,0)$ to \eqref{perturbedAbstractPrimal} evaluated at $(x,0)$, that is
\begin{equation}\label{1251}
\begin{array}{rcl}
\lim\limits_{|\Delta x| + |\delta| \to 0}\frac{1}{|\Delta x| + |\delta|}\left|
\widehat \xi(x + \Delta x, \delta) - \frac{\widehat\xi(x,0)}{x}\left(x + \Delta x(1 + \alpha(x,0)) + \delta \alpha_{d}(x,0) \right)\frac{1}{L^{\delta}} \right|&=& 0,\\
\end{array}
\end{equation}
where the convergence takes place in $\mathbb P$-probability. Likewise, let $\beta$ and $\beta_{d}$, which are solutions to \eqref{byy} and \eqref{bdd}, correspondingly, are the partial derivatives of the solution $\widehat\eta(y,0)$ to \eqref{perturbedAbstractDual} evaluated at $(y,0)$, where $y = u_x(x,0)$, in the sense that 
\begin{equation}\label{1252}
\begin{array}{rcl}
\lim\limits_{|\Delta y| + |\delta |\to 0}\frac{1}{|\Delta y| + |\delta|}\left|
\widehat \eta(y + \Delta y, \delta) - \frac{\widehat\eta(y,0)}{y}\left(y + \Delta y(1 + \beta(y,0)) + \delta \beta_{d}(y,0) \right){L^{\delta}} \right|&=& 0,\\
\end{array}
\end{equation}
where the convergence takes place in $\mathbb P$-probability. 
\end{thm}
From Theorem \ref{134}, we obtain the following Corollary.
\begin{cor} Under the conditions of Theorem \ref{134}, \eqref{1251} is equivalent to 
\begin{equation}\nonumber
\begin{array}{rcl}
\lim\limits_{|\Delta x| + |\delta| \to 0}\frac{1}{|\Delta x| + |\delta|}\left|
 \xi(x + \Delta x, \delta) - \xi(x,0)- \frac{\xi(x,0)}{x}\left(\Delta x(\alpha(x,0) + 1) + \delta (\alpha_{d}(x,0) + xF)\right) \right|&=& 0.\\
\end{array}
\end{equation}
Likewise, \eqref{1252} holds if and only if
\begin{equation}\nonumber
\begin{array}{rcl}
\lim\limits_{|\Delta y| + |\delta |\to 0}\frac{1}{|\Delta y| + |\delta|}\left|
 \eta(y + \Delta y, \delta) - \eta(y,0)-\frac{\eta(y,0)}{y}\left(\Delta y(\beta(y,0) + 1) + \delta (\beta_{d}(y,0) - yF) \right)\right|&=& 0,\\
\end{array}
\end{equation}
where the convergence takes place in $\mathbb P$-probability. 
\end{cor}
\section*{Proofs} We begin the proofs with technical lemmas.
\begin{lem}\label{11301}
Let Assumption \ref{rra} hold and ${d}\in\left(\max\left(\exp(-1/c_2), \exp(-c_1)\right), 1\right]$. Then for every $x>0$, we have
\begin{displaymath}\begin{array}{rcl}
U'({d} x) &\leq& \frac{1}{1 + c_2\log({d})}U'(x),\\
-V'({d} x) &\leq& \frac{1}{1 +\frac{1}{c_1}\log({d})}(-V'(x)).\\
\end{array}\end{displaymath}
\end{lem}
\begin{proof}
Let us fix an arbitrary $x>0$ and  ${d}\in\left(\max\left(\exp(-1/c_2), \exp(-c_1)\right), 1\right]$. Then using Assumption \ref{rra} and monotonicity of $U'$, we get
\begin{displaymath}
\begin{array}{rcl}
U'({d} x) - U'(x) &=& \int_{{d}}^1(-U''(tx))xdt \\ 
&=& \int_{{d}}^1(-U''(tx))tx\frac{dt}{t} \\ 
&\leq& c_2\int_{{d}}^1U'(tx)\frac{dt}{t} \\ 
&\leq& c_2U'({d} x)(-\log({d})). \\ 
\end{array}
\end{displaymath}
Therefore, we obtain 
\[ U'({d} x) (1 + c_2\log({d})) \leq U'(x),\]
This implies the first assertion of the lemma. The other one can be shown entirely similarly.
\end{proof}
\begin{cor}\label{11302}
Under the conditions of Lemma \ref{11301}, for every $k\in\mathbb N,$ we have
\begin{displaymath}\begin{array}{rcl}
 U'({d}^k x)&\leq& \frac{1}{(1 + c_2\log ({d}))^k}U'(x),\\
  -V'({d}^k x)&\leq& \frac{1}{(1 + \frac{1}{c_1}\log ({d}))^k}(-V'(x)).\\
 \end{array}\end{displaymath}
\end{cor}
Below 
$1_{E}$ denotes the indicator function of a set $E$.
\begin{lem}\label{11303}
Let Assumption \ref{rra} holds. Then for every $z\in(0,1]$ and $x>0$, we have
\begin{displaymath}\begin{array}{rcl}
U'(zx)&\leq &z^{-c_2}U'(x),\\
-V'(zx)&\leq &z^{-\frac{1}{c_1}}(-V'(x)).\\
\end{array}\end{displaymath}
\end{lem}
\begin{proof} Let us fix an arbitrary ${d} \in (\exp(-1/c_2), 1)$. Using monotonicity of $U'$ and Corollary \ref{11302}, for every $z\in(0,1]$ and $x>0$, we get
\begin{equation}\label{1211}
\begin{array}{rcl}
U'(zx) &=& \sum\limits_{k = 1}^\infty U'(zx) 1_{\{z\in({d}^k,{d}^{k-1}]\}} \\
&\leq& \sum\limits_{k = 1}^\infty U'({d}^k x) 1_{\{z\in({d}^k,{d}^{k-1}]\}} \\
&\leq& U'(x)\sum\limits_{k = 1}^\infty \frac{1}{\left(1 + c_2\log({d}) \right)^k} 1_{\{z\in({d}^k,{d}^{k-1}]\}}.\\
\end{array}
\end{equation}
Let us set $$a_1({d}) \triangleq \frac{1}{1 + c_2\log({d}) } >1\quad and\quad  
a_2({d}) \triangleq \frac{\log\left(1 + c_2\log({d})\right)}{\log({d})} =- \frac{\log(a_1({d}))}{\log({d})}>0.$$
As $a_1({d})>1$ and  for every $k\in\mathbb N$
$${d}^k<z\leq {d}^{k-1}
\quad is~equivalent~to\quad\frac{\log(z)}{\log({d})}<k\leq \frac{\log(z)}{\log({d})} + 1,$$
we deduce that for every $z\in(0,1]$, we have
\begin{equation}\label{1221}
\begin{array}{rcl}
\frac{1}{\left(1 + c_2\log({d}) \right)^k}1_{\{z\in({d}^k,{d}^{k-1}]\}}& \leq & a_1({d})a_1({d})^{\frac{\log(z)}{\log({d})}}1_{\{z\in({d}^k,{d}^{k-1}]\}}
\\
&=&a_1({d})\left(a_1({d})^{\frac{1}{\log({d})}}\right)^{\log(z)}1_{\{z\in({d}^k,{d}^{k-1}]\}}
\\
&=&  a_1({d})z^{-a_2({d})}1_{\{z\in({d}^k,{d}^{k-1}]\}}.
\end{array}
\end{equation}
Plugging \eqref{1221} in \eqref{1211}, we get
\begin{displaymath}
U'(zx) \leq U'(x)\sum\limits_{k = 1}^\infty a_1({d})z^{-a_2({d})}1_{\{z\in({d}^k,{d}^{k-1}]\}} =  a_1({d})z^{-a_2({d})}U'(x),\quad for~every~z\in(0,1]~and~x>0.
\end{displaymath}
As $\lim\limits_{{d} \uparrow 1}a_1({d}) = 1$ and 
\begin{displaymath}
\lim\limits_{{d} \uparrow 1}a_2({d}) = \lim\limits_{{d} \uparrow 1}\frac{\log(1 + c_2\log({d}))}{\log({d})} = \lim\limits_{y\uparrow 0}\frac{\log(1 + c_2y)}{y} = \lim\limits_{y\uparrow 0}\frac{c_2}{1 + c_2y} = c_2, 
\end{displaymath}
taking the limit in the latter inequality, we obtain that
\begin{displaymath}
U'(zx) \leq \lim\limits_{{d}\uparrow 1}  a_1({d})z^{-a_2({d})}U'(x) = z^{-c_2}U'(x),
\end{displaymath}
for every $z\in(0,1]$ and $x>0$. 
The other assertion can be proven similarly. 
This completes the proof of the lemma.
\end{proof}
\begin{cor}\label{rraCor}
Under Assumption \ref{rra}, for every $z>0$ and $x>0$, we have
\begin{displaymath}\begin{array}{rcccl}
U'(zx) &\leq& \max\left(z^{-c_2},1\right)U'(x)&\leq& \left(z^{-c_2}+1\right)U'(x),\\
-V'(zx) &\leq& \max\left(z^{-\frac{1}{c_1}},1\right)(-V'(x))&\leq &\left(z^{-\frac{1}{c_1}}+1\right)(-V'(x)).\\
\end{array}\end{displaymath}
\end{cor}

\section*{Proof of the second-order expansion}
\begin{lem}\label{12132}
Let $x>0$ be fixed and the conditions of Theorem \ref{byproduct} hold, and $y=u_x(x,0)$.
For arbitrary random variables $\alpha^0$ and $\alpha^1$ in $\mathcal A^\infty(x,0)$, let us define 
\begin{equation}\label{12131}
\begin{array}{rcl}
\psi(s,t) &\triangleq& \frac{1}{x}\left(x + s(1 + \alpha^0) + t\alpha^1\right)\frac{1}{L^t},\\
w(s,t) &\triangleq& \mathbb E\left[U(\xi \psi(s,t))\right],\quad (s,t)\in\mathbb R^2,\\
\end{array}
\end{equation}
where $\xi = \widehat \xi(x,0)$ is the solution to \eqref{perturbedAbstractPrimal} corresponding to $x>0$ and $\delta = 0$.
Then $w$ admits the following second-order expansion at $(0,0)$.
\begin{equation}\label{12133}
w(s,t) = w(0,0) +(s\quad t) \nabla w(0,0) + \tfrac{1}{2}(s\quad t) 
H_w
\begin{pmatrix}
s\\
t\\
\end{pmatrix} + o(s^2 + t^2),
\end{equation}
where \textcolor{black}{
\begin{displaymath}
\begin{array}{rcl}
w_s(0,0)&=&u_x(x,0),\\
w_t(0,0)&=&xy\mathbb E^{\mathbb R(x,0)}\left[ F\right],\\
\end{array}
\end{displaymath}}
and
\begin{displaymath}
H_w \triangleq \begin{pmatrix} 
  w_{ss}(0,0)     & w_{st}(0,0)\\ 
  w_{st}(0,0) 	& w_{tt}(0,0)\\ 
\end{pmatrix},
\end{displaymath}
 where the second-order partial derivatives of $w$ at $(0,0)$ are given by 
 \begin{displaymath}\begin{array}{rcl}
 w_{ss}(0,0) & = &-\frac{y}{x}\mathbb E^{\mathbb R(x)}\left[ A(\xi)(1 + \alpha^0)^2\right],
 \\
 w_{st}(0,0) & = &-\frac{y}{x}\mathbb E^{\mathbb R(x)}\left[A(\xi)(1 + \alpha^0)(xF + \alpha^1) - xF(1 + \alpha^0) \right],\\
 w_{tt}(0,0) & = &-\frac{y}{x}\mathbb E^{\mathbb R(x)}\left[A(\xi)(\alpha^1 + xF)^2 - 
2xF\alpha^1 - x^2(F^2 + G) \right].\\
 \end{array}\end{displaymath}
\end{lem}
\begin{proof} As $\alpha^0$ and $\alpha^1$ are in $\mathcal A^\infty$, there exists constant $\varepsilon \in(0,1)$, such that
\begin{equation}\label{11510}
|\alpha^0| + |\alpha^1|\leq \frac{x}{6\varepsilon} - 1, \quad \mathbb P-a.s.
\end{equation}
Let us fix an arbitrary $(s,t)\in B_{\varepsilon}(0,0)$ and define
\begin{displaymath}
\widetilde \psi (z) \triangleq \psi(zs, zt),\quad z\in (-1,1).
\end{displaymath}
Note that
\begin{equation}\label{1151}
\frac{2}{3}\leq \widetilde \psi(z) L^{zt}\leq \frac{4}{3}, \quad z\in (-1,1).
\end{equation}
As $$\psi_t(s, t) = \frac{\alpha^1}{xL^t} + \psi(s,t) \left(F+ tG\right)\quad and \quad 
\psi_s(s, t) = \frac{1 + \alpha^0}{xL^t},$$
we get
\begin{equation}\label{1155}
\begin{array}{c}
\widetilde \psi'(z) = \psi_s(sz, tz)s + \psi_t(sz, tz)t =\frac{1 + \alpha^0}{xL^{zt}}s + \left(\frac{\alpha^1}{xL^{zt}} + \widetilde \psi(z) \left(F+ ztG\right)\right)t.\\
\end{array}
\end{equation}
Similarly, since 
\begin{equation}\nonumber\begin{array}{rcl}
\psi_{tt}(s, t) &=& \frac{2\alpha^1}{xL^{t}}\left(F+ tG\right) + \psi(s,t) \left(\left(F+ tG\right)^2 + G\right),\\
\psi_{st}(s, t) &=& \frac{1 + \alpha^0}{xL^{t}}\left(F+tG\right),\quad and\quad\psi_{ss}(s, t) = 0,\\
\end{array}\end{equation}
we obtain
\begin{displaymath}
\begin{array}{rcl}
\widetilde \psi''(z) &=& \psi_{tt}(zs, zt)t^2 + 2\psi_{st}(zs, zt)ts + \psi_{ss}(zs, zt)s^2\\
 &=& \left(\frac{2\alpha^1}{xL^{zt}}\left(F+ ztG\right) + \tilde\psi(z) \left(\left(F+ ztG\right)^2 + G\right) \right)t^2 + 2 \frac{1 + \alpha^0}{xL^{zt}}\left(F+ztG\right)ts.\
\end{array}
\end{displaymath}
Setting $W(z) \triangleq U(\xi \widetilde \psi(z))$, $z\in(-1, 1)$, by direct computations, we get
\begin{equation}\label{1152}\begin{array}{rcl}
W'(z) &=& U'(\xi \widetilde \psi(z))\xi \widetilde \psi'(z), \\
W''(z) &=& U''(\xi \widetilde \psi(z))\left(\xi \widetilde \psi'(z)\right)^2 + U'(\xi \widetilde \psi(z))\xi \widetilde \psi''(z).\\
\end{array}
\end{equation}
Let us define 
$$a_2 \triangleq 2^{c_2 + 2}\quad and \quad J \triangleq 1+|F| + G.$$
From \eqref{1155} using \eqref{11510} and \eqref{1151}, we get
 $$|\widetilde \psi'(z) |\leq 2J\exp(\varepsilon J),\quad  \widetilde \psi(z) ^{-c_2} + 1\leq 2^{c_2+ 1} \exp(c_2 \varepsilon J),\quad z\in(-1, 1).$$
Therefore, from \eqref{1152} using Corollary \ref{rraCor}, we obtain
\begin{equation}\label{1133}
\begin{array}{rcl}
\sup\limits_{z\in (-1, 1)}|W'(z) | &\leq & \sup\limits_{z\in (-1, 1)}U'(\xi)\xi\left( (\widetilde \psi(z))^{-c_2}+1\right)\left|\widetilde \psi'(z)\right| \\
&\leq &a_2  U'(\xi)\xi J\exp((c_2 + 1)\varepsilon J)\\
&\leq &a_2  U'(\xi)\xi J\exp(a_2 \varepsilon J).\\
\end{array}
\end{equation}
Similarly, from \eqref{1152} applying  Assumtion \ref{rra} and Corrollary \ref{rraCor}, we deduce the existence of a constant $a_3 > 0$, such that  
\begin{equation}\label{1141}
\begin{array}{rcl}
\sup\limits_{z\in (-1, 1)}|W''(z) |&\leq &
a_3 U'(\xi)\xi J^2\exp(a_3\varepsilon J).\\
\end{array}
\end{equation}
Combining \eqref{1133} and \eqref{1141}, we obtain
\begin{equation}\nonumber
\sup\limits_{z\in (-1, 1)}\left(|W'(z) | + |W''(z) |\right)\leq
U'(\xi)\xi\left(
a_2   J\exp(a_2\varepsilon J)  +
a_3 J^2\exp(a_3\varepsilon J)\right).
\end{equation}
Consequently, as $1\leq J\leq J^2$, by setting $a_1 \triangleq \max(a_2, a_3)$, for every $z_1$ and $z_2$ in $(-1, 1)$, we get
\begin{equation}\label{1157}\begin{array}{rcl}
\left|\frac{W(z_1) - W(z_2)}{z_1 -z_2}\right| + \left|\frac{W'(z_1) - W'(z_2)}{z_1 -z_2}\right|&\leq& 
4a_1 U'(\xi)\xi
  J^2\exp(a_1\varepsilon J).\\
\end{array}\end{equation}
By passing to a smaller $\varepsilon$, if necessary, and by applying H\"older's inequality, we deduce from Assumption \ref{integrabilityAs} that the right-hand side of \eqref{1157} integrable. As the right-hand side of \eqref{1157} depends on $\varepsilon$ (and not on $(s,t)$), the assertion of the lemma follows from the dominated convergence theorem.

\end{proof}
\begin{cor}\label{cor12131}
 Let let $x>0$ be fixed, the conditions of Theorem \ref{byproduct} hold, and $y=u_x(x,0)$. Then, we have
\begin{equation}\label{12134}
u(x+\Delta x,\delta) \geq u(x,0)+ \Delta x y+  \delta xy\mathbb E^{\mathbb R(x,0)}\left[ F\right]+ \tfrac{1}{2}(\Delta x\quad \delta) 
H_u(x,0)
\begin{pmatrix}
\Delta x\\
\delta\\
\end{pmatrix} + o(\Delta x^2 + \delta^2),
\end{equation}
where  $H_u(x,0)$ is given by \eqref{12136}.
\end{cor}
\begin{proof}
The result follows from Lemma \ref{12132} via the approximation of the solutions to \eqref{axx} and \eqref{add}, which are the elements of $\mathcal A^2(x,0)$, by the elements of $\mathcal A^\infty(x,0)$.
\end{proof}
Similarly to Lemma \ref{12132} and Corollary \ref{cor12131}, we can establish the following results.
\begin{lem}\label{157}
Let $x>0$ be fixed, the conditions of Theorem \ref{byproduct} hold, and $y = u_x(x,0)$.
For arbitrary random variables $\beta^0$ and $\beta^1$ in $\mathcal B^\infty(y,0)$, let us define 
\begin{equation}\nonumber
\begin{array}{rcl}
\phi(s,t) &\triangleq& \frac{1}{y}\left(y + s(1 + \beta^0) + t\beta^1\right){L^t},\\
\bar w(s,t) &\triangleq& \mathbb E\left[V(\eta \phi(s,t))\right],\quad (s,t)\in\mathbb R^2,\\
\end{array}
\end{equation}
where $\eta = \widehat \eta(y,0)$ is the solution to \eqref{perturbedAbstractDual} corresponding to $y>0$ and $\delta = 0$. Then at $(0,0)$, $\bar w$ admits the following second-order expansion
\begin{equation}\nonumber
\bar w(s,t) = \bar w(0,0) +(s\quad t) \nabla \bar w(0,0) + \tfrac{1}{2}(s\quad t) 
H_{\bar w}
\begin{pmatrix}
s\\
t\\
\end{pmatrix} + o(s^2 + t^2),
\end{equation}
where \textcolor{black}{
\begin{displaymath}
\begin{array}{rcl}
\bar w_s(0,0)&=&v_y(y,0),\\
\bar w_t(0,0)&=&xy\mathbb E^{\mathbb R(x,0)}\left[ F\right],\\
\end{array}
\end{displaymath}}
and
\begin{displaymath}
H_{\bar w} \triangleq \begin{pmatrix} 
  \bar w_{ss}(0,0)     & \bar w_{st}(0,0)\\ 
  \bar w_{st}(0,0) 	& \bar w_{tt}(0,0)\\ 
\end{pmatrix},
\end{displaymath}
 where the second-order partial derivatives of $\bar w$ at $(0,0)$ are given by 
 \begin{displaymath}\begin{array}{rcl}
 \bar w_{ss}(0,0) & = &\frac{x}{y}\mathbb E^{\mathbb R(x,0)}\left[ B(\eta)(1 + \beta^0)^2\right],
 \\
 \bar w_{st}(0,0) & = &\frac{x}{y}\mathbb E^{\mathbb R(x,0)}\left[B(\eta)(1 + \beta^0)(-yF + \beta^1) +yF(1 + \beta^0) \right],\\
 \bar w_{tt}(0,0) & = &\frac{x}{y}\mathbb E^{\mathbb R(x,0)}\left[B(\eta)(\beta^1 - yF)^2 +
2yF\beta^1 - y^2(F^2 - G) \right].\\
 \end{array}\end{displaymath}
\end{lem}

\begin{lem}\label{lem12131}
 Let let $x>0$ be fixed, the conditions of Theorem \ref{byproduct} hold, and $y = u_x(x,0)$. Then, we have
\begin{equation}\label{12135}
v(y+\Delta y,\delta) \leq v(y,0)-\Delta y x +\delta xy\mathbb E^{\mathbb R(x,0)}\left[ F\right]+ \tfrac{1}{2}(\Delta y\quad \delta) 
H_v(y,0)
\begin{pmatrix}
\Delta y\\
\delta\\
\end{pmatrix} + o(\Delta y^2 + \delta^2),
\end{equation}
where  $H_v(y,0)$ is given by \eqref{12137}.
\end{lem}

\section*{Closing the duality gap}
We begin from the proof of Theorem \ref{closingGapThm}.
\begin{proof}[Proof of Theorem \ref{closingGapThm}]
It follows from \cite[Lemma 2]{KS2006} that
\begin{equation}\label{12241}
\begin{array}{rcl}
A(\xi)(1 + \alpha) &=& a(x,x)(1 + \beta),\\
B(\eta)(1 + \beta) &=& b(y,y)(1 + \alpha).\\
\end{array}
\end{equation}
Using standard techniques from calculus of variations, we can show that the solutions to \eqref{add} and \eqref{bdd} satisfy
\begin{equation}\label{12242}
\begin{array}{rcl}
A(\xi)(\alpha_{d} + xF) - xF &=& c + \tilde{ \beta},\\
B(\eta)(\beta_{d} - yF) + yF &=& d +\tilde{ \alpha},\\
\end{array}
\end{equation}
where $\tilde{ \beta}\in\mathcal B^2(y,\delta)$, $\tilde{ \alpha}\in\mathcal A^2(x,\delta)$, and  $c$ and $d$ are some constants. We will characterize $\tilde{ \beta}$, $\tilde{ \alpha}$, and $d$ below. 
Let us set 
\begin{equation}\label{12244}
\tilde {\tilde{ \alpha}} \triangleq \tilde{ \alpha} - d \alpha\in\mathcal A^2(x,0).
\end{equation} 
It follows from the second equation in \eqref{12242} that 
\begin{equation}\begin{array}{rcl}
\beta_{d} - yF & =&A(\xi)\left( d - yF+ \tilde{ \alpha} \right)\\
& =&A(\xi)\left( d +d \alpha - yF+ \tilde{ \alpha} - d\alpha\right)\\
& =& d a(x,x) (1 + \beta)+ A(\xi)\left(- yF+ \tilde{\tilde{ \alpha}}\right),\\
\end{array}\end{equation}
where we have used \eqref{12241} in the last equality. Multiplying by $-\frac{x}{y}$, we obtain
\begin{displaymath}\begin{array}{rcl}
A(\xi)\left(xF - \frac{x}{y}\tilde{\tilde{ \alpha}}\right) &= &
-\frac{x}{y}(\beta_{d} - yF) -\frac{x}{y} d a(x,x) (1 + \beta),\\
\end{array}
\end{displaymath}
and thus 
\begin{displaymath}
A(\xi) \left(xF - \frac{x}{y}\tilde{\tilde{ \alpha}} \right) -xF= \tilde d + \tilde{\tilde{ \beta}},
\end{displaymath}
where
\begin{displaymath}
\tilde d = -\frac{x}{y}d a(x,x)\in\mathbb R\quad and \quad 
\tilde {\tilde{ \beta}} = -\frac{x}{y} d a(x,x) \beta  - \frac{x}{y}\beta_{d} \in\mathcal B^2(y,0).
\end{displaymath}
It follows from characterization of the unique solution to \eqref{add} given by \eqref{12242} that $$- \frac{x}{y}\tilde {\tilde{ \alpha}} = \alpha_{d},\quad equivalently \quad
\tilde{\tilde{ \alpha}} = -\frac{y}{x} \alpha_{d}.$$
From \eqref{12244}, we obtain
\begin{displaymath}
\tilde{ \alpha} = \tilde{\tilde{ \alpha}} + d\alpha = -\frac{y}{x}\alpha_{d} + d \alpha.
\end{displaymath}
Plugging this back into the second equality in \eqref{12242}, we get
\begin{displaymath}
B(\eta)(\beta_{d} - yF) = d(1 + \alpha) - \frac{y}{x}(\alpha_{d} + xF).
\end{displaymath}
Multiplying by $\frac{x}{y}A(\xi)$, we obtain
\begin{equation}\label{12246}
A(\xi) (\alpha_{d} + xF) = \frac{x}{y}d a(x,x) (1 + \beta) - \frac{x}{y}(\beta_{d} - yF).
\end{equation}
Setting $d' \triangleq \frac{x}{y}d a(x,x) 
$, we claim that
\begin{equation}\label{12245}
d' = a(x,d),
\end{equation}
where $a(x,d)$ is defined in \eqref{axd}.
Multiplying both sides of \eqref{12246} by $(1 + \alpha)$, taking expectation under $\mathbb R(x,0)$, and using orthogonality of the elements of $\mathcal A^2(x,0)$ and $\mathcal B^2(y,0)$, we get 
\begin{displaymath}
\begin{array}{rcl}
\mathbb E^{\mathbb R(x,0)}\left[A(\xi)(\alpha_{d} + xF)(1 + \alpha) \right]&=& d'\mathbb E^{\mathbb R(x,0)}\left[(1 + \beta)(1 + \alpha)\right] \\
&&- \frac{x}{y}\mathbb E^{\mathbb R(x,0)}\left[(\beta_{d} - yF)(1 + \alpha)\right]\\
&=& d'+\mathbb E^{\mathbb R(x,0)}\left[ xF(1 + \alpha)\right].\\
\end{array}
\end{displaymath}
Therefore,
\begin{displaymath}
d' = \mathbb E^{\mathbb R(x,0)}\left[A(\xi)(\alpha_{d} + xF)(1 + \alpha) -xF(1 + \alpha)  \right] = a(x,d),
\end{displaymath}
where in the last equality we have used \eqref{axd}. 
Thus, \eqref{12245} holds. Now, \eqref{12246} with $ \frac{x}{y}d a(x,x)  = a(x,d)$ and \eqref{12242} prove \eqref{keyGap2}. \eqref{12251} can be shown similarly. 
As $A(\xi) = \frac{1}{B(\eta)}$, \eqref{keyGap2} and \eqref{12251} imply \eqref{12154}. 

It remains to prove \eqref{keyGap}.
Let us set
\begin{displaymath}\begin{array}{rclcrcl}
\bar \beta &\triangleq & \beta + 1,& \quad & \bar \alpha & \triangleq & \alpha + 1,\\
\bar \beta_{d} &\triangleq & \beta_{d} - yF,& \quad & \bar \alpha_{d} & \triangleq & \alpha_{d} + xF.\\
\end{array}
\end{displaymath}
Then from \eqref{add} using \eqref{keyGap2}, we get
\begin{equation}\label{12248}\begin{array}{rcl}
\frac{y}{x} a(\delta,\delta) &= &
\mathbb E^{\mathbb R(x, 0)}\left[\frac{y}{x}a(x,d) \bar \beta\bar \alpha_{d} - \bar \beta_{d} \bar \alpha_{d} \right] \\
&&- \frac{y}{x}\mathbb E^{\mathbb R(x, 0)}\left[2xF \alpha_{d}\right] - xy\mathbb E^{\mathbb R(x, 0)}\left[F^2 + G\right].\\
\end{array}
\end{equation}
Likewise, from \eqref{bdd} via \eqref{12248}, we obtain
\begin{equation}\label{12249}
\frac{x}{y}b(d,d) = \mathbb E^{\mathbb R(x, 0)}\left[\frac{x}{y}b(y,d) \bar \alpha\bar \beta_{d} - \bar \beta_{d} \bar \alpha_{d} + 2\beta_{d} xF -xy(F^2 -G)\right].
\end{equation}
Let us define
\begin{equation}\nonumber
T_1 \triangleq \mathbb E^{\mathbb R(x, 0)}\left[
\frac{y}{x}a(x,d) \bar \beta\bar \alpha_{d} +\frac{x}{y}b(y,d) \bar \alpha\bar \beta_{d} \right],
\end{equation}
and 
\begin{equation}\nonumber
T_2 \triangleq \mathbb E^{\mathbb R(x, 0)}\left[
-2\bar \beta_{d} \bar \alpha_{d} - 2yF\alpha_{d} + 2 xF\beta_{d} - 2 xy F^2
 \right].
\end{equation}
Then, adding \eqref{12248} and \eqref{12249}, we deduce that
\begin{equation}\label{122410}
\begin{array}{rcl}
\frac{y}{x} a(\delta,\delta) +\frac{x}{y}b(d,d) &=&
T_1 + T_2.\\
\end{array}
\end{equation}
Let us rewrite $T_2$ as 
\begin{equation}\label{122412}
\begin{array}{rcl} T_2 &=& 
\mathbb E^{\mathbb R(x, 0)}\left[
-2\bar \beta_{d} \bar \alpha_{d} - 2yF\alpha_{d} + 2 xF\beta_{d} - 2 xy F^2\right]
\\
&=&
\mathbb E^{\mathbb R(x, 0)}\left[-2(\beta_{d} - yF)(\alpha_{d} + xF) - 2 yF\alpha_{d} + 2xF\beta_{d} - 2 xyF^2\right]
\\
&=& \mathbb E^{\mathbb R(x, 0)}\left[
-2\beta_{d}\alpha_{d} - 2\beta_{d}xF + 2 yF\alpha_{d} + 2 xyF^2 
 - 2 yF\alpha_{d} \right.\\
 &&\hspace{13mm}\left.+ 2xF\beta_{d} - 2 xyF^2\right]
= \mathbb E^{\mathbb R(x, 0)}\left[
-2\beta_{d}\alpha_{d} \right] = 0,
\\
\end{array} 
\end{equation}
as all the terms under the expectation cancel except for $-2\beta_{d}\alpha_{d}$, which has still $0$ expectation by orthogonality of $\mathcal A^2(x,0)$ and $\mathcal B^2(y,0)$.
Let us consider $T_1$. First, from \eqref{12154}, we get
\begin{equation}
\frac{x}{y}b(y,d)= \frac{a(x,d)}{a(x,x)} = a(x,d)b(y,y).
\end{equation}
Therefore, we can rewrite $T_1$ as 
\begin{equation}\label{122414}\begin{array}{rcl}
T_1 &=& \mathbb E^{\mathbb R(x,0)}\left[ \frac{y}{x}a(x,d) \bar \beta \bar \alpha_{d} + a(x,d)b(y,y)\bar \alpha \bar \beta_{d}\right]\\
&=& a(x,d) \mathbb E^{\mathbb R(x,0)}\left[ \frac{y}{x}\bar \beta \bar \alpha_{d} +b(y,y)\bar \alpha \bar \beta_{d}\right].\\
\end{array}
\end{equation}
On the other hand, from \eqref{byd} we can express $b(y,d)$ in terms of $\bar \beta$, $\bar \beta_{d}$, $\bar \alpha$, and $\bar \alpha_{d}$ as follows.
\begin{equation}\label{122411}
b(y,d) =
\mathbb E^{\mathbb R(x,0)}\left[B(\eta)\bar \beta_{d}\bar \beta + \frac{y}{x}\bar \beta \bar \alpha_{d}\right] =
\mathbb E^{\mathbb R(x,0)}\left[b(y,y)\bar \alpha\bar \beta_{d} + \frac{y}{x}\bar \beta \bar \alpha_{d}\right],
\end{equation}
where in the last equality we have used \eqref{12241}. Comparing \eqref{122411} with  \eqref{122414}, we get 
$$T_1 = a(x,d)b(y,d).$$
Plugging this into \eqref{122410} and using \eqref{122412}, we deduce that 
$$\frac{y}{x} a(d,d) +\frac{x}{y}b(d,d) =  a(x,d)b(y,d),$$
i.e. \eqref{keyGap} holds. This completes the proof of the lemma.
\end{proof}

\begin{lem}\label{12212}
Let $x>0$ be fixed, the conditions of Theorem \ref{byproduct} hold, and $y = u_x(x,0)$. Then, for
\begin{equation}\label{12215}
\Delta y =-\frac{y}{x b(y,y)}\left(\frac{x}{y}b(y,d)\delta +  \Delta x\right),
\end{equation}
 we have
\begin{equation}\label{12213}
\left(\begin{array}{c}\Delta y \\ \delta\end{array}\right)^T H_v(y,0) \left(\begin{array}{c}\Delta y \\ \delta\end{array}\right) + 2\Delta x\Delta y = 
\left(\begin{array}{c}\Delta x \\ \delta\end{array}\right)^T H_u(x,0) \left(\begin{array}{c}\Delta x \\ \delta\end{array}\right).
\end{equation}
\end{lem}
\begin{proof}
First, note that $b(y,y)>0$ in  \eqref{12215}.
By direct computations, proving \eqref{12213} is equivalent to establishing the following equality.
\begin{equation}\label{12214}
-\frac{y}{x b(y,y)} \left( \frac{x}{y}b(y,d)\delta +  \Delta x\right)^2 = 
\left(\begin{array}{c}\Delta x \\ \delta\end{array}\right)^T H_u(x,0) \left(\begin{array}{c}\Delta x \\ \delta\end{array}\right) 
- \frac{x}{y}b(d,d)\delta^2.
\end{equation}
Now, let us consider the right-hand side or \eqref{12214}. By direct computations, it can be rewritten as follows.
\begin{equation}\label{12216}\begin{array}{c}
 -\frac{y}{x}\Delta x^2 a(x,x) + 2\Delta x \delta \left(- \frac{y}{x} a(x,d)\right) 
- \delta^2\left(\frac{y}{x}a(d,d) + \frac{x}{y}b(d,d)\right)=\\ 
 -\frac{{y}}{{x}b(y,y)}\Delta x^2
+ 2\Delta x \delta \left(- \frac{y}{x} a(x,d)\right) 
- \delta^2a(x,d)b(y,d),\\ 
\end{array}\end{equation}
where the last equality follows from \eqref{122111} and \eqref{keyGap}. 
We deduce from \eqref{12154} that
\begin{equation}\label{122110}
a(x,d) = \frac{x}{y}\frac{b(y,d)}{b(y,y)}.
\end{equation}
Pluggin \eqref{122110} into  \eqref{12216}, we can rewrite the right-hand side of \eqref{12216} as 
\begin{displaymath}\begin{array}{c}
-\frac{{y}}{{x}b(y,y)}\Delta x^2 
- 2\Delta x \delta \frac{b(y,d)}{b(y,y)}
- \delta^2\frac{x}{y}\frac{(b(y,d))^2}{b(y,y)}
=-\frac{y}{x b(y,y)}\left(\Delta x  + \frac{x}{y}b(y,d)\delta \right)^2,
\end{array}\end{displaymath}
which is precisely the left-hand side of \eqref{12214}. We have just shown that \eqref{12214} holds. By the argument preceding \eqref{12214}, this implies that \eqref{12213} is valid as well. This completes the proof of the lemma.
\end{proof}

\begin{lem}\label{5161}
Let $x>0$ be fixed, the conditions of Theorem \ref{byproduct} hold, and $y = u_x(x,0)$.  Then, we have
\begin{equation}\label{5162}
u(x+\Delta x,\delta) = u(x,0)+ \Delta x y+  \delta xy\mathbb E^{\mathbb R(x,0)}\left[ F\right]+ \tfrac{1}{2}(\Delta x\quad \delta) 
H_u(x,0)
\begin{pmatrix}
\Delta x\\
\delta\\
\end{pmatrix} + o(\Delta x^2 + \delta^2),
\end{equation}
where  $H_u(x,0)$ is given by \eqref{12136}.
Likewise
\begin{equation}\label{5163}
v(y+\Delta y,\delta) = v(y,0)-\Delta y x +\delta xy\mathbb E^{\mathbb R(x,0)}\left[ F\right]+ \tfrac{1}{2}(\Delta y\quad \delta) 
H_v(y,0)
\begin{pmatrix}
\Delta y\\
\delta\\
\end{pmatrix} + o(\Delta y^2 + \delta^2),
\end{equation}
where  $H_v(y,0)$ is given by \eqref{12137}.
\end{lem}
\begin{proof}
For small $\Delta x$  and $\delta$ and with $\Delta y$ given by \eqref{12215}, we get
 from conjugacy of $u$ and $v$ and Lemma \ref{lem12131} that
\begin{equation}\label{5-9-1}
\begin{array}{rcl}
 u(x+\Delta x, \delta) &\leq&
 v(y+\Delta y, \delta) + (x + \Delta x)(y + \Delta y)\\
 &\leq& v(y, 0)-\Delta y x +\delta xy\mathbb E^{\mathbb R(x,0)}\left[ F\right]
 +\tfrac{1}{2}\left(\begin{array}{c}\Delta y \\ \delta\end{array}\right)^T H_v(y,0) \left(\begin{array}{c}\Delta y \\ \delta\end{array}\right)\\
 &&+ xy + y\Delta x + x \Delta y + \Delta x\Delta y + o(\Delta y^2+\delta^2),\\
\end{array}
\end{equation}
where 
$
 H_v(y,0)
$
is given in \eqref{12137}.
As $y = u_x(x, 0)$ and $x = -v_y(y, 0)$, 
collecting terms in the right-hand side of (\ref{5-9-1}), we obtain

\begin{equation}\label{5-9-2}
 \begin{array}{rcl}
  u(x + \Delta x, \delta) &\leq& u(x, 0)+ \Delta x y+  \delta xy\mathbb E^{\mathbb R(x,0)}\left[ F\right]\\
  &&+\tfrac{1}{2}\left(\begin{array}{c}\Delta y \\ \delta\end{array}\right)^T H_v(y,0) \left(\begin{array}{c}\Delta y \\ \delta\end{array}\right) + \Delta x\Delta y + o(\Delta x^2 + \delta^2).
 \end{array}
\end{equation}
Likewise, using Corollary \ref{cor12131}
, we get
\begin{equation}\label{122113}
u(x+\Delta x,\delta) \geq u(x,0)+ \Delta x y+  \delta xy\mathbb E^{\mathbb R(x,0)}\left[ F\right]+ \tfrac{1}{2}(\Delta x\quad \delta) 
H_u(x,0)
\begin{pmatrix}
\Delta x\\
\delta\\
\end{pmatrix} + o(\Delta x^2 + \delta^2).
\end{equation}
By Lemma \ref{12212}
, the quadratic terms in \eqref{5-9-2} and \eqref{122113} are equal. Therefore, \eqref{5-9-2} and \eqref{122113} imply that $u$ admits a second-order expansion at $(x,0)$ given by \eqref{5162}.
Similarly we can prove 
\eqref{5163}.

\end{proof}
\begin{proof}[Proof of Theorem \ref{byproduct}]
The assertions of Theorem \ref{byproduct} follow from Lemma \ref{5161}.
\end{proof}
\begin{proof}[Proof of Theorem \ref{seconOrderThmA}]
Expansions \eqref{122121} and \eqref{122122} follow from Lemma \ref{5161} and Theorem \ref{byproduct}.
\end{proof}
\section*{Derivatives of the optimizers}
We begin with a technical lemma.
\begin{lem}\label{131}
Let $x>0$ be fixed, the conditions of Theorem \ref{byproduct} hold, $y=u_x(x,0)$, and let $(\delta^n)_{n\geq 1}$ be a sequence, which converges  to $0$. Then, we have
\begin{displaymath}
\lim\limits_{n\to\infty}\mathbb E\left[V\left( \widehat {\eta}(y,0)L^{\delta^n}\right) \right] = v(y,0).
\end{displaymath}
\end{lem}
\begin{proof}
The proof goes along the lines of the proof of Lemma \ref{12132}, it is therefore skipped.
\end{proof}
\begin{lem}
Let $x>0$ be fixed, the conditions of Theorem \ref{byproduct} hold, $y=u_x(x,0)$, and $(y^n, \delta^n)_{n\in\mathbb N}$ be a sequence, which converges to $(y,0)$. Then $\eta^n \triangleq \widehat {\eta}(y^n, \delta^n)$, $n\geq 1$, converges to $\eta \triangleq \widehat {\eta}(y, 0)$ in probability and $V(\eta^n)$, $n\geq 1$, converges to $V(\eta)$ in $\mathbf L^1(\mathbb P).$
\end{lem}
\begin{proof} In view of Theorem \ref{byproduct}, without loss of generality, we may assume that $v(y^n, \delta^n)$ is finite for every $n\in\mathbb N$.
Let us assume by contradiction that $\left( \eta ^n\right)_{n\in\mathbb N}$ does not converge in probability to $\eta$. Then there exists $\varepsilon >0$, such that 
\begin{equation}\nonumber
\limsup\limits_{n\to \infty}\mathbb P\left[\left| \eta^n - \eta\right| >\varepsilon\right] >\varepsilon.
\end{equation}
Let us define $\widetilde {\eta}^n \triangleq \frac{\eta^n}{L^{\delta^n}}$, $n\geq 1$, and $\bar y\triangleq \sup\limits_{n\geq 1}y^n$. As $\left(\widetilde {\eta}^n \right)_{n\in\mathbb N}\subset \mathcal D(\bar y, 0)$ and $\left(L^{\delta^n}\right)_{n\in\mathbb N}$ converges to $1$ in probability (therefore, in particular $\left(L^{\delta^n}\right)_{n\in\mathbb N}$ is bounded in $\mathbf L^0$), by possibly passing to a smaller $\varepsilon$, we may assume that 
\begin{equation}\nonumber
\limsup\limits_{n\to \infty} \mathbb P\left[\left| \eta^n - \eta\right| >\varepsilon,~\left| \widetilde {\eta}^n - \eta \right|L^{\delta^n}\leq \tfrac{1}{\varepsilon}\right] >\varepsilon.
\end{equation}
Let us define 
\begin{displaymath}
h^n \triangleq \frac{1}{2}\left( \widetilde {\eta}^n + \eta\right)L^{\delta^n}
=\frac{1}{2}\left(  \eta^n + \eta L^{\delta^n}\right)\in \mathcal D\left(\tfrac{y_n + y}{2}, \delta^n\right),\quad n\geq 1.
\end{displaymath}
From convexity of $V$, we have
\begin{equation}\nonumber
V(h^n) \leq \frac{1}{2} \left(V(\eta^n) + V\left(\eta L^{\delta^n}\right) \right),
\end{equation}
and from the strict convexity of $V$, we deduce the existence of a positive constant $\varepsilon_0$, such that 
\begin{equation}\nonumber
\limsup\limits_{n\to\infty}\mathbb P\left[ V(h^n) \leq \frac{1}{2} \left(V(\eta^n) + V\left(\eta L^{\delta^n}\right) \right) - \varepsilon_0\right]>\varepsilon_0.
\end{equation}
Therefore, using Lemma \ref{131}, we obtain
\begin{equation}\label{132}
\begin{array}{rcl}
\limsup\limits_{n\to\infty}\mathbb E\left[ V(h^n)\right] &\leq &
\tfrac{1}{2} \limsup\limits_{n\to\infty}\mathbb E\left[V\left( \eta ^n\right) \right] + \tfrac{1}{2} \limsup\limits_{n\to\infty}\mathbb E\left[V\left( \eta L^{\delta^n}\right) \right] - \varepsilon_0^2\\
&=&\tfrac{1}{2}\limsup\limits_{n\to\infty}v(y^n, \delta^n) + \tfrac{1}{2} v(y,0) - \varepsilon_0^2 \\
&=& v(y,0) - \varepsilon_0^2,\\
\end{array}
\end{equation}
where in the last equality we have also used continuity of $v$ at $(y,0)$, which follows from Theorem \ref{byproduct}.
On the other hand, as $h^n\in\mathcal D\left(\tfrac{y_n + y}{2}, \delta^n\right)$, $n\geq 1$, we get 
\begin{equation}\label{133}
 \limsup\limits_{n\to\infty}v\left( \tfrac{y_n + y}{2}, \delta^n\right)\leq \limsup\limits_{n\to\infty}\mathbb E\left[ V(h^n)\right].
\end{equation}
Combining \eqref{132} and \eqref{133} and using continuity of $v$ at $(y,0)$ again, we get
\begin{displaymath}
v(y,0) = \limsup\limits_{n\to\infty}v\left( \tfrac{y_n + y}{2}, \delta^n\right)\leq \limsup\limits_{n\to\infty}\mathbb E\left[ V(h^n)\right] \leq v(y,0) - \varepsilon_0^2,
\end{displaymath}
which is a contradiction as $\varepsilon_0 \neq 0$. Thus, $(\eta^n)_{n\in\mathbb N}$ converges to $\eta$ in probability. In turn, this and continuity of $v$ at $(y,0)$ imply the other assertion of the lemma.
\end{proof}
\begin{proof}[Proof of Theorem \ref{134}] We will only prove \eqref{1252}, as \eqref{1251} can be shown similarly. In view of Theorem \ref{byproduct}, without loss of generality we will assume that for every $n\in\mathbb N$, $u(\cdot, \delta^n)$ and $v(\cdot, \delta^n)$ are finite-valued functions. The rest of the proof goes along the lines of the proof of Theorem 2 in \cite{KS2006}. 
Let $(y^n, \delta^n)_{n\in\mathbb N}$ be a sequence, which converges to $(y,0)$, where $y = u_x(x,0)>0$. Let $\widehat \eta^n = \widehat \eta(y^n, \delta^n)$, $n\in\mathbb N$, denote the corresponding dual optimizers and set
\begin{displaymath}\begin{array}{rclc}
\phi_1 &\triangleq& \frac{1}{2}\min\left(\widehat \eta(y, 0), \inf\limits_{n\in\mathbb N} \widehat \eta^n\right)>0,&\mathbb P-a.s.\\
\phi_2 &\triangleq& 2\max\left(\widehat \eta(y, 0), \sup\limits_{n\in\mathbb N} \widehat \eta^n\right)<\infty,&\mathbb P-a.s.\\
\theta &\triangleq& \inf\limits_{\phi_1\leq t\leq \phi_2}V''(t).&
\end{array}\end{displaymath}
Note that the construction of $\phi_1$ and $\phi_2$ implies that $\theta >0$ , $\mathbb P-a.s.$. Let us also fix $\beta^0$ and $\beta^1$ in $\mathcal B^\infty(y,0)$ and define
\begin{displaymath}
\eta^n \triangleq \frac{\widehat \eta(y,0)}{y}\left(y + \Delta y^n (\beta^0 + 1) + \delta^n \beta^1  \right)L^{\delta^n}\in\mathcal D(y^n, \delta^n),\quad n\in\mathbb N,
\end{displaymath}
where $\Delta y^n \triangleq y^n - y$. 
As $\beta^0$ and $\beta^1$ are bounded, without loss of generality we will assume that 
\begin{displaymath}
\frac{1}{2} \widehat \eta(y,0) \leq \eta^n \leq 2 \widehat \eta(y,0),\quad n\in\mathbb N,
\end{displaymath}
which implies that $$\phi_1 \leq\eta^n\leq \phi_2.$$
Using the definition of $\theta$, we get
\begin{equation}\label{141}
V\left(\eta^n\right) - V\left( \widehat \eta^n\right) \geq V'(\widehat \eta^n) \left(\eta^n - \widehat \eta^n\right) + \theta \left(\eta^n - \widehat \eta^n \right)^2.
\end{equation}
By \cite[Theorem 3.2]{Mostovyi2015}, $-V'(\widehat \eta^n) = \widehat \xi(x^n, \delta^n)$ is the optimal solution to \eqref{perturbedAbstractPrimal} at $x^n = -v_y(y^n, \delta^n)$, such that 
\begin{displaymath}
\mathbb E\left[\widehat \xi(x^n, \delta^n) \widehat \eta^n \right] = x^ny^n.
\end{displaymath} 
Moreover, the bipolar construction of the sets $\mathcal C(x^n, \delta^n)$ and $\mathcal D(y^n, \delta^n)$ implies that 
\begin{displaymath}
\mathbb E\left[\widehat \xi(x^n, \delta^n) \eta^n \right] \leq x^ny^n.
\end{displaymath}
Therefore, we obtain
\begin{displaymath}
\mathbb E\left[ V'(\widehat \eta^n)\left(\eta^n - \widehat \eta^n\right)\right]\geq 0.
\end{displaymath}
Combining this with \eqref{141}, we get
\begin{equation}\label{142}
\mathbb E\left[ \theta \left(\eta^n - \widehat \eta^n\right)^2\right] \leq \mathbb E\left[V(\eta^n) \right] - v(y^n, \delta^n).
\end{equation}
From \textcolor{black}{Lemma \ref{157}}, we deduce
\begin{displaymath}\begin{array}{rcl}
\mathbb E\left[V(\eta^n) \right] &=& \mathbb E\left[V(\eta^n) \right] \\
&=& v(y,0) -x\Delta y^n + v_\delta(y,0)\delta^n + \tfrac{1}{2}(\Delta y^n\quad \delta^n) H_{\bar w}\begin{pmatrix}\Delta y^n\\ \delta^n\end{pmatrix}  + o((\Delta y^n)^2 + (\delta^n)^2).
\end{array}\end{displaymath}
Combining this with \eqref{142} and using the expansion for $v$ from Theorem \ref{seconOrderThmA}, we obtain
\begin{equation}\label{144}
\limsup\limits_{n\to\infty}\frac{1}{(\Delta y^n)^2 + (\delta^n)^2 }\left(\mathbb E\left[\eta^n\right] - v(y^n, \delta^n) \right)\leq \frac{1}{2}\Vert H_{\bar w} - H_v(y,0)\Vert,
\end{equation}
where for a vector $a=\begin{pmatrix}a_1\\ a_2\end{pmatrix}$ and a two-by-two matrix $A$, we define their norms as  $$\Vert a\Vert \triangleq \sqrt{a_1^2 + a_2^2} \quad and \quad \Vert A \Vert \triangleq \sup\limits_{a\in\mathbb R^2}\frac{\Vert Aa\Vert }{\Vert a\Vert}.$$ In view of Lemma \ref{157} (by the choice of $\beta^0$ and $\beta^1$), we can make the right-hand side of \eqref{144} arbitrarily small. Combining this with \eqref{142},  we deduce that 
\begin{displaymath}
\limsup\limits_{n\to\infty}\frac{1}{(\Delta y^n)^2 + (\delta^n)^2 }\mathbb E\left[ \theta \left(\eta^n - \widehat \eta^n\right)^2\right]
\end{displaymath}
can also be made arbitrarily small.
 As $\theta >0$, $\mathbb P-a.s.$, the assertion of the theorem follows.
\end{proof}


\section{Proofs of Theorems \ref{mainThm1}, \ref{mainThm2}, \ref{162}, \ref{161}, and \ref{thmCorOptimizer}}\label{secProofC}

In order to link abstract theorems to their concrete counterparts, we will have to establish some structural properties of the perturbed primal and dual admissible sets first.
\section*{Characterization of primal and dual admissible sets}
The following lemma gives a useful characterization of the primal and dual admissible sets after perturbations.
\begin{lem}\label{keyConcreteCharLem}
Under Assumption \eqref{NUPBR}, 
for every $\delta \in\mathbb R$, we have
\begin{displaymath}\begin{array}{rcl}
\mathcal Y(1,\delta) &=& \mathcal Y(1,0)\mathcal E\left(-\delta\nu\cdot S^0\right),\\
\mathcal X(1,\delta) &=& \mathcal X(1,0)\frac{1}{ \mathcal E\left(-\delta\nu\cdot S^0\right)}.\\
\end{array}\end{displaymath}
\end{lem}
\begin{proof}
Let us fix $\delta \in\mathbb R$. Then, for an arbitrary predictable and $S^\delta$-integrable process $\pi$, let $X^\delta \triangleq \mathcal E\left(\pi\cdot S^\delta\right)$. Then $X^\delta \in\mathcal X(1, \delta)$. Let us consider $X^0 \triangleq X^\delta \mathcal E\left(-\delta\nu\cdot S^0\right)$. One can see that $X^0\in\mathcal X(1,0)$. The remainder of the proof is straightforward, it is therefore skipped.
\end{proof}
\begin{proof}[Proof of Theorem \ref{mainThm1}]
Condition \eqref{NUPBR} implies that the respective closures of the convex solid hulls of $\left\{X_T:~X\in\mathcal X(1,0)\right\}$ and $\left\{Y_T:~Y\in\mathcal Y(1,0)\right\}$ satisfy (abstract) Assumption \ref{bipolar}. In view of Lemma \ref{keyConcreteCharLem}, we have
$$\left\{\frac{X_T}{L^\delta}:~X\in\mathcal X(1,0)\right\} = \left\{{X_T}:~X\in\mathcal X(1,\delta)\right\},$$ likewise $$\left\{Y_TL^\delta:~Y\in\mathcal Y(1,0)\right\} = \left\{Y_T:~Y\in\mathcal Y(1,\delta)\right\},\quad \delta\in\mathbb R.$$ Therefore, the respective closures of convex solid hulls of
$$\left\{{X_T}:~X\in\mathcal X(1,\delta)\right\}\quad
  and \quad
  \left\{Y_T:~Y\in\mathcal Y(1,\delta)\right\}$$
  satisfy abstract condition \eqref{bipolarL}.  The relationship between (abstract) Assumption \ref{integrabilityAs} and  Assumption \ref{integrabilityAssumption} is apparent. It remains to show that the sets $\mathcal M^2(x)$ and $\mathcal N^2(x)$ satisfy (abstract) Assumption \ref{complimentarity}. However this follows from continuity of $S^0$ and \cite[Lemma 6]{KS2006}. Therefore, the assertions of Theorem \ref{mainThm1} follow from (abstract) Theorem \ref{byproduct}. 
\end{proof}
\begin{proof}[Proof of Theorem \ref{mainThm2}]
As in the proof of Theorem \ref{mainThm2}, the assertions of Theorem \ref{mainThm2} follow from (abstract) Theorem \ref{seconOrderThmA}.
\end{proof}
\begin{proof}[Proof of Theorem \ref{162}]
Similarly to the proof of Theorem \ref{mainThm2}, the assertions of Theorem \ref{162} follow from (abstract) Theorem \ref{closingGapThm}.
\end{proof}
\begin{proof}[Proof of Theorem \ref{161}]
As above, the affirmations of this theorem follow from (abstract) Theorem \ref{134}.
\end{proof}
For the proof of Theorem \ref{thmCorOptimizer}, we will need the following technical lemma. First, for $(\delta, \Delta x, \varepsilon)\in \mathbb R\times(-x,\infty)\times (0,\infty)$, let us set 
\begin{equation}\label{5191}
f(\delta, \Delta x, \varepsilon)\triangleq \frac{u(x,0) + (\Delta x\quad \delta) \nabla u(x,0) + \tfrac{1}{2}(\Delta x\quad \delta) 
H_u(x,0)
\begin{pmatrix}
\Delta x\\
\delta\\
\end{pmatrix} - \mathbb E\left[U\left(X^{\Delta x, \delta,  \varepsilon}_T\right) \right]}{\Delta x^2 + \delta^2},
\end{equation}
where 
$\nabla u(x,0)$, $H_u(x,0)$, and  $X^{\Delta x, \delta, \varepsilon}$'s are defined in  \eqref{gradientC}, \eqref{12136C}, and \eqref{4232}, respectively.
\begin{lem}\label{51913}
Assume that $x>0$ is fixed and the assumptions of Theorem \ref{mainThm1} hold.  Then, for $f$ defined in \eqref{5191}, there exists a monotone function $g$, such that 
\begin{equation}\label{51911}
g(\varepsilon)\geq\lim\limits_{|\Delta x|+|\delta|\to 0}f(\delta, \Delta x, \varepsilon),\quad \varepsilon >0,
\end{equation}
and \begin{equation}\label{51912}
\lim\limits_{\varepsilon \to 0}g(\varepsilon) = 0.
\end{equation}

\end{lem}
\begin{proof} The proof goes along the lines of the proof of Lemma \ref{12132}. We only outline the main steps for brevity of exposition. 
For a fixed $\varepsilon >0$, let us define 
\begin{equation}\label{12131C}
\begin{array}{rcl}
\psi(\Delta x,\delta) &\triangleq& \frac{x+\Delta x}{x}\exp\left((\Delta x\gamma^{0,\varepsilon} + \delta\gamma^{1,\varepsilon})\cdot M^R_T - \tfrac{1}{2}(\Delta x\gamma^{0,\varepsilon} + \delta\gamma^{1,\varepsilon})^2\cdot \langle M\rangle_T\right)\frac{1}{L^\delta},\\
w(\Delta x,\delta) &\triangleq& \mathbb E\left[U(\widehat X_T(x,0)\psi(\Delta x,\delta))\right],\quad (\Delta x,\delta)\in\mathbb R^2,\\
\end{array}
\end{equation}
where $M^R$ is defined in \eqref{4231}. 
Let us first fix $\varepsilon' > 0$, 
then fix $(\Delta x,\delta)\in B_{\varepsilon'}(0,0)$, and set
\begin{displaymath}
\widetilde \psi (z) \triangleq \psi(z\Delta x, z\delta),\quad z\in (-1,1).
\end{displaymath}
By direct computations, we get
\begin{equation}\label{51916}
\widetilde \psi'(z) = \psi_{\Delta x}(z\Delta x, z\delta)\Delta x + \psi_\delta(z\Delta x, z\delta)\delta,
\end{equation}
where 
 \begin{displaymath}\begin{array}{c}
\psi_{\Delta x}(\Delta x, \delta) = \psi(\Delta x, \delta)\left( \frac{1}{x + \Delta x} + (\Delta x\gamma^{0,\varepsilon}\cdot M^R_T - ((\Delta x\gamma^{0,\varepsilon} + \delta\gamma^{1,\varepsilon})\gamma^{0,\varepsilon})\cdot \langle M\rangle_T)\right),\\
\psi_\delta(\Delta x, \delta) = \psi(\Delta x, \delta)\left( \gamma^{1,\varepsilon}\cdot M^R_T - ((\Delta x\gamma^{0,\varepsilon} + \delta\gamma^{1,\varepsilon})\gamma^{1,\varepsilon})\cdot \langle M\rangle_T + F +\delta G\right),
\end{array}
\end{displaymath}
where $F$ and $G$ are defined in \eqref{defFG}. Similarly, we obtain
\begin{displaymath}
\widetilde \psi''(z) = \psi_{\Delta x\Delta x}(z\Delta x, z\delta)\Delta x^2 + 2\psi_{\Delta x \delta}(z\Delta x, z\delta)\Delta x\delta + \psi_{\delta \delta}(z\Delta x, z\delta )\delta^2,
\end{displaymath}
where
\begin{equation}\nonumber\begin{array}{rcl}
\psi_{\Delta x\Delta x}(\Delta x, \delta) &=&\psi(\Delta x, \delta)\left( \frac{1}{x + \Delta x} + (\Delta x\gamma^{0,\varepsilon}\cdot M^R_T - ((\Delta x\gamma^{0,\varepsilon} + \delta\gamma^{1,\varepsilon})\gamma^{0,\varepsilon})\cdot \langle M\rangle_T)\right)^2\\
&& + \psi(\Delta x, \delta)\left(\gamma^{0,\varepsilon}\cdot M^R_T+(\gamma^{0,\varepsilon})^2\cdot \langle M\rangle_T -\frac{1}{(x + \Delta x)^2}\right),\\
\psi_{\Delta x \delta}(\Delta x,\delta) &=& \psi(\Delta x, \delta)\left( \frac{1}{x + \Delta x} + (\Delta x\gamma^{0,\varepsilon}\cdot M^R_T - ((\Delta x\gamma^{0,\varepsilon} + \delta\gamma^{1,\varepsilon})\gamma^{0,\varepsilon})\cdot \langle M\rangle_T)\right)\times\\
&&\times \left( \gamma^{1,\varepsilon}\cdot M^R_T - ((\Delta x\gamma^{0,\varepsilon} + \delta\gamma^{1,\varepsilon})\gamma^{1,\varepsilon})\cdot \langle M\rangle_T + F +\delta G\right)\\
&&+ \psi(\Delta x, \delta) \left((\gamma^{1,\varepsilon}\gamma^{0,\varepsilon})\cdot \langle M\rangle_T\right),\\
\psi_{\delta \delta}(\Delta x, \delta) &=&\psi(\Delta x, \delta)\left( \gamma^{1,\varepsilon}\cdot M^R_T - ((\Delta x\gamma^{0,\varepsilon} + \delta\gamma^{1,\varepsilon})\gamma^{1,\varepsilon})\cdot \langle M\rangle_T + F +\delta G\right)^2 \\
&& + \psi(\Delta x, \delta)\left((\gamma^{1,\varepsilon})^2\cdot  \langle M\rangle_T  + G \right).\\
\end{array}\end{equation}
Setting $W(z) \triangleq U(\widehat X_T(x,0)\widetilde \psi(z))$, $z\in(-1, 1)$, by direct computations, we get
\begin{equation}\nonumber
\begin{array}{rcl}
W'(z) &=& U'(\widehat X_T(x,0) \widetilde \psi(z))\widehat X_T(x,0) \widetilde \psi'(z), \\
W''(z) &=& U''(\widehat X_T(x,0)\widetilde \psi(z))\left(\widehat X_T(x,0) \widetilde \psi'(z)\right)^2 + U'(\widehat X_T(x,0)\widetilde \psi(z))\widehat X_T(x,0)\widetilde \psi''(z).\\
\end{array}
\end{equation}
As in Lemma \ref{12132}, from boundedness of $\gamma^{0,\varepsilon}\cdot M^R_T$, $\gamma^{1,\varepsilon}\cdot M^R_T$, $(\gamma^{0,\varepsilon})^2\cdot \langle M\rangle_T$, and $(\gamma^{1,\varepsilon})^2\cdot \langle M\rangle_T$, via Corollary \ref{rraCor} and Assumption \ref{integrabilityAs}, one can show that
\begin{equation}\nonumber
\left|\frac{W(z_1) - W(z_2)}{z_1 -z_2}\right| + \left|\frac{W'(z_1) - W'(z_2)}{z_1 -z_2}\right| \leq \eta,
\end{equation}
for some random variable $\eta$, which depend on $\varepsilon'$ and which is integrable for a sufficiently small $\varepsilon'$.
By direct computations, the  derivatives of $W$ plugged inside the expectation lead to  the ``exact'' gradient $\nabla u(x,0)$  and the ``approximate'' 
Hessian $H^{\varepsilon}_u(x,0)$. This results in \eqref{51911}. Now, approximation by $\varepsilon \rightarrow 0$ leads to
$H^{\varepsilon} _u(x,0)\rightarrow H_u(x,0)$, and, therefore we obtain \eqref{51912}. Finally, one can choose $g$ to be monotone.
\end{proof}
\begin{proof}[Proof of Theorem \ref{thmCorOptimizer}]
First, for $f$ defined in \eqref{5191}, via Lemma \ref{51913}, we deduce the existence of a monotone function $g$, such that \eqref{51911} and \eqref{51912} hold.
Let us define
$$\phi(\varepsilon) \triangleq \left\{(\delta, \Delta x):~f(t\delta, t\Delta x, \varepsilon)\leq 2 g(\varepsilon),~~for~every~t\in[0,1]\right\},\quad \varepsilon>0,$$
$$r(\varepsilon) \triangleq \tfrac{1}{2}\sup\left\{r\leq \varepsilon:~B_r(0,0)\subseteq \phi(\varepsilon)\right\}, \quad \varepsilon>0.$$
Note that $r(\varepsilon)>0$ for every $\varepsilon>0$. With
$$\varepsilon(\delta, \Delta x) \triangleq \inf\left\{ \varepsilon:~ r(\varepsilon) \geq \sqrt{\Delta x^2 + \delta^2}\right\}, \quad (\delta, \Delta x)\in\mathbb R\times(-x,\infty),$$
we have
$$\lim\limits_{|\Delta x|+|\delta|\to 0}\frac{u(x + \Delta x, \delta) - \mathbb E\left[U\left(X^{\Delta x, \delta, \varepsilon(\delta, \Delta x)}_T\right) \right]}{\Delta x^2 + \delta^2} = 0.$$
\end{proof}

\section{Counterexample}\label{counterexamples}
In the following example we show that even when $0$-model is nice, but Assumption \ref{integrabilityAssumption} fails, we might have $$u(z,\delta) = v(z,\delta) = \infty\quad for~every~\delta \neq 0~and~z>0.$$
\begin{exa}\label{2142}
Consider the $0$-model, where 
$$T = 1,\quad M = B, \quad \lambda \equiv 1, \quad and\quad  U(x) = \frac{x^p}{p},\quad p\in(0,1).$$
Let assume that $B$ is a Brownian motion defined on a filtered probability space $\left(\Omega, \mathcal F, \mathbb P\right)$, where the filtration $\left(\mathcal F_t\right)_{t\in[0,T]}$ is generated by $B$. We recall that for the utility function $U(x) = \frac{x^p}{p}$, the convex conjugate is $V(y) = \frac{y^{-q}}{q}$, where $q = \frac{p}{1-p}$. Let $Z^0$ denote the martingale deflator for $S^0$.
The direct computations yield
\begin{displaymath}
\mathbb E\left[(Z^0_1)^{-q} \right] = \exp\left( \tfrac{1}{2}q(q+1)\right)\in\mathbb R.
\end{displaymath}
Therefore by \cite{KS2003}, the standard conclusions of the utility maximization theory hold. The primal and dual optimizers are 
$$\widehat X_1(x,0)=x\exp \left((q+1)B_1 + \tfrac{1}{2}(1-q^2) \right)\quad and\quad \widehat Y_1(y,0) = y\exp\left( -B_1 - \tfrac{1}{2}\right).$$

Now, let us consider a process $\nu$ such that 
\begin{equation}\label{2141}\nu\cdot B_1 = B_1^3,\quad \mathbb P-a.s.
\end{equation}
Let us denote $I_t \triangleq t$, $t\in[0,1]$.
As  $$\frac{d\mathbb R(x,0)}{d\mathbb P} = \exp\left(-\frac{q(q+1)}{2}\right)\exp\left(qB_1 + q\tfrac{1}{2} \right) = \exp\left(qB_1 - \frac{q^2}{2} \right),\quad x>0,$$ with notations \eqref{defFG}, for every $c>0$, we get
\begin{equation}\label{222}
\begin{array}{rcl}
\mathbb E^{\mathbb R(x,0)}\left[ \exp\left(c\left(\left|F \right| +  G\right)\right)\right] 
&=& \mathbb E\left[ \exp\left(qB_1 - \frac{q^2}{2} \right) \exp\left(c\left| \nu\cdot B_1 + \nu\cdot I_1\right| + c\nu^2\cdot I_1\right)\right]\\
&=& \mathbb E\left[\exp\left(qB_1 - \frac{q^2}{2} + c\left| B^3_1 + \nu\cdot I_1\right| + c\nu^2\cdot I_1\right)\right]\\
&\geq& \mathbb E\left[\exp\left(qB_1 - \frac{q^2}{2} + c |B^3_1|  -c |\nu|\cdot I_1 + c\nu^2\cdot I_1\right)\right]\\
&\geq& \exp\left( - \frac{q^2}{2}  - \tfrac{c}{4}\right)\mathbb E\left[\exp\left(qB_1 + c |B^3_1|   +c\left(|\nu|-\tfrac{1}{2}\right)^2\cdot I_1\right)\right]\\
&\geq& \exp\left( - \frac{q^2}{2} - \tfrac{c}{4}\right)\mathbb E\left[\exp\left(qB_1 + c |B^3_1|\right)\right]\\
&=& \exp\left( - \frac{q^2}{2} - \tfrac{c}{4}\right)\frac{1}{\sqrt{2\pi}}\int_{\mathbb R}\exp\left(qy + c |y^3| - y^2/2\right)dy = \infty,\\
\end{array}
\end{equation}
i.e. Assumption \ref{integrabilityAssumption} does not hold. 

For every $\delta\in\mathbb R$, we can express the local martingale deflator $Z^\delta$ as follows
$$Z^\delta_t = \exp\left(-(\lambda + \delta\nu)\cdot B_t - \tfrac{1}{2}(\lambda +\delta\nu)^2 \cdot I_t\right),\quad t\in[0,1].$$
For $p\in(0,1)$, as $q>p>0$, we have
\begin{equation}\nonumber
\mathbb E\left[(Z^\delta_1)^{-q} \right] =\mathbb E\left[\exp\left( q (\lambda  + \delta\nu)\cdot B_1 + \frac{q}{2}(\lambda + \delta\nu)^2\cdot I_1\right) \right]
\geq 
\mathbb E\left[\exp\left(  q (\lambda + \delta\nu)\cdot B_1 \right) \right].
\end{equation}
Therefore, using \eqref{2141}, we get 
\begin{displaymath}
\begin{array}{rcl}
\mathbb E\left[(Z^\delta_1)^{-q} \right]&\geq &
\mathbb E\left[\exp\left( q (\lambda  + \delta\nu)\cdot B_1 \right) \right] \\
&=& 
\mathbb E\left[\exp\left(  q (B_1 +  \delta B^3_1 ) \right) \right] \\
&=& 
\int_{\mathbb R}\frac{1}{\sqrt{2\pi}} \exp\left( -\frac{y^2}{2} + q(y + \delta y^3)\right)dy = \infty,\\
\end{array}
\end{displaymath}
for every $\delta \neq 0$. Consequently, $v(1, \delta) = \infty$ for every $\delta \neq 0$. Moreover, one can find a constant $D>0$, such that 
\begin{equation}\nonumber
\begin{array}{rcl}
u(1,\delta) & \geq & \mathbb E\left[ U\left(\widehat X^0_1(1,0) \exp\left( \delta F + \tfrac{1}{2}\delta^2G\right)\right)\right] \\
& = & D\mathbb E\left[\exp \left(qB_1 +\tfrac{q}{2}\right) \exp\left( p\delta \nu\cdot B_1 + p\delta \nu\cdot I_1 + \tfrac{p}{2}\delta^2\nu^2\cdot I_1\right)\right] \\
& = & D\mathbb E\left[\exp \left(qB_1 + p\delta B^3_1  +\tfrac{q-p}{2} +  \tfrac{p}{2} \left(\delta\nu + 1\right)^2\cdot I_1\right)\right].\\
\end{array}
\end{equation}
As $(q-p)$ and $\tfrac{p}{2} \left(\delta\nu + 1\right)^2\cdot I_1$ are nonnegative, we get
\begin{equation}\nonumber
\begin{array}{rcl}
u(1,\delta) & \geq &D\mathbb E\left[\exp \left(qB_1 + p\delta B^3_1\right)\right] \\
& = &D\frac{1}{\sqrt{2\pi}}\int_{\mathbb R}
\exp\left(qy + p\delta y^3 - y^2/2\right)dy = \infty,\\
\end{array}
\end{equation}
for every $\delta \neq  0$. 
\end{exa}

\section{Relationship to the risk-tolerance wealth process}\label{secRiskTol}
Following \cite{KS2006b}, we recall that for an initial wealth $x>0$ and $\delta\in\mathbb R$, the {\it risk-tolerance wealth process} is a maximal wealth process $R(x,\delta)$, such that 
\begin{equation}\label{1251riskTol}
R_T(x,\delta) = -\frac{U'(\widehat X_T(x,\delta))}{U''(\widehat X_T(x,\delta))},
\end{equation}
i.e. it is a replication process for the random payoff given by the right-hand side of \eqref{1251riskTol}. In general the risk-tolerance wealth process $R(x, \delta)$ may not exist. It is shown in \cite{KS2006b} that the existence of the risk-tolerance wealth process is closely related to some important properties of the marginal utility-based prices and to the validity of the second-order expansions of the value functions under the presence of random endowment. Below we establish a relationship between the existence of $R(x, 0)$ and the second-order expansions of the value functions in the present context.

The following theorem is a version of \cite[Theorem 4]{KS2006b}. Despite the fact that the assertions of \cite[Theorem 4]{KS2006b} are obtained under the existence of an equivalent martingale measure assumption in \cite{KS2006b}, the proof goes through also under condition \eqref{NUPBR}, no changes are needed. Therefore, the proof of the following theorem is not presented.
\begin{thm}\label{riskTolThmLocal}
Let $x>0$ be fixed, assume that \eqref{NUPBR}, \eqref{finCond}, and Assumption \ref{rra} hold, and denote $y = u_x(x,0)$. Then the following assertions are equivalent:
\begin{enumerate}
\item The risk-tolerance wealth process $R(x, 0)$ exists.
\item The value function $u$ admits the expansion \eqref{122121} at $(x,0)$ and 
$u_{xx}(x,0) = -\frac{y}{x}a(x,x)$ satisfies
\begin{equation}\nonumber
\frac{\left( u_x(x,0)\right)^2}{u_{xx}(x,0)} = \mathbb E\left[\frac{\left(U'(\widehat X_T(x, 0) \right)^2}{U''(\widehat X_T(x,0))} \right],
\end{equation}
\begin{equation}\nonumber
u_{xx}(x,0) = \mathbb E\left[U''(\widehat X_T(x, 0) \left(\frac{R_T(x,0)}{R_0(x,0)}\right)^2\right].
\end{equation}
\item The value function $v$ admits the expansion \eqref{122122} at $(y,0)$ and $v_{yy}(y,0)=\frac{x}{y}b(y,y)$ satisfies
\begin{equation}
y^2 v_{yy}(y,0) = \mathbb E\left[ \left(\widehat Y_T(y,0)\right)^2V''(\widehat Y_T(y,0)) \right] = xy \mathbb E^{\mathbb R(x,0)}\left[ B(\widehat Y_T(y,0))\right].
\end{equation}
\end{enumerate}
In addition, if these assertions are valid, then the initial value of $R(x)$ is given by 
\begin{equation}\label{1253}
R_0(x, 0) = -\frac{u_x(x,0)}{u_{xx}(x,0)} = \frac{x}{a(x,x)},
\end{equation}
the product $R(x,0)Y(y,0) = \left( R_t(x,0)Y_t(y,0)\right)_{t\in [0,T]}$ is a uniformly integrable martingale and 
\begin{equation}\label{1255}
\lim\limits_{\Delta x \to 0}\frac{\widehat X_T(x + \Delta x, 0) - \widehat X_T(x , 0)}{\Delta x} = \frac{R_T(x,0) }{R_0(x,0)},
\end{equation}
\begin{equation}\label{1256}
\lim\limits_{\Delta y \to 0}\frac{\widehat Y_T(y + \Delta y, 0) - \widehat Y_T(y, 0)}{\Delta y} = \frac{\widehat Y_T(y,0)}{y},
\end{equation}
where the limits in \eqref{1255} and \eqref{1256} take place in $\mathbb P$-probability.
\end{thm}

As in \cite{KS2006b}, for $x>0$ and with $y = u_x(x,0)$, let us define
\begin{equation}\label{tildeR}
\frac{d\widetilde {\mathbb R}(x,0)}{d \mathbb P} \triangleq \frac{R_T(x,0) \widehat Y_T(y,0)}{R_0(x,0) y},
\end{equation}
and choose $ \frac{R(x,0)}{R_0(x,0)}$ as a num\'eraire, i.e., let us set
\begin{equation}\label{SR}
S^{R(x,0)} \triangleq \left(\frac{R_0(x,0)}{ R(x,0)}, \frac{ R_0(x,0)S}{ R(x,0)}\right).
\end{equation}
We define the spaces of martingales \begin{equation}\label{eq:tildeM}
\mathcal {\widetilde M}^2(x,0) \triangleq \left\{ M\in \mathbf H_0^2(\mathbb{\widetilde R}(x,0)):~M = H\cdot S^{R(x,0)}\right\},
\end{equation} and 
$\mathcal{\widetilde N}^2(y,0)$ it the orthogonal complement in $\mathbf H_0^2(\mathbb{\widetilde R}(x,0))$.  We start with the following simple lemma (stated without a proof) relating the change of num\'eraire  to the structure of martingales:
\begin{lem}\label{lem:change-rt} 
Let $x>0$ be fixed, assume that the conditions of Theorem \ref{riskTolThmLocal} hold, and denote $y = u_x(x,0)$. Then, we have
\begin{equation}\label{3261}
M \in\mathcal M^2(x,0)\quad {{ if~and~only~if}} \quad M\frac{\widehat X_T(x,0)}{R_T(x,0)}\in\mathcal {\widetilde M}^2(x,0),
\end{equation}
and
$$N\in \mathcal{ N}^2(y,0  ) \quad {{ if~and~only~if}} \quad N\in  \mathcal{\widetilde N}^2(y,0).$$\end{lem}
The following theorem describes the structural properties  the approximations in Theorems \ref{mainThm2}, \ref{162}, and \ref{161}  under the assumption that the risk-tolerance process exists. In words, \emph{the second order approximation of the value function optimal strategies amounts to a Kunita-Watanabe decomposition} under the changes of measure and num\'eraire described above.
\begin{thm}\label{riskTolThm2}
Let $x>0$ be fixed, assume that the conditions of Theorem \ref{riskTolThmLocal} hold, and denote $y = u_x(x,0)$. Let us also assume that the risk-tolerance process $R(x,0)$ exists.
Consider the Kunita-Watanabe decomposition of  the square integrable martingale 
$$P_t \triangleq \mathbb E^{\mathbb{\widetilde R}(x,0)}\left[\left(A(\widehat X_T(x,0)) -1\right) xF|\mathcal F_t\right],\ \ t\in[0,T]$$ 
given by
\begin{equation}\label{eq:Kunita-Watanabe}
P=P_0-{\widetilde M}^1-{\widetilde N}^1, \quad{where}\quad {\widetilde M}^1 \in \mathcal {\widetilde M}^2(x,0),\quad {\widetilde N}^1  \in \mathcal{\widetilde N}^2(y,0),\quad P_0\in \R.
\end{equation}
Then, the optimal solutions $M^1(x,0)$ and $N^1(y,0)$ of the quadratic optimization problems  \eqref{addC} and \eqref{bddC} can be obtained from the Kunita-Watanabe decomposition \eqref{eq:Kunita-Watanabe} by reverting to the original num\'eraire, according to Lemma \ref{lem:change-rt}, through the identities 
\begin{equation}\label{3251}
{\widetilde M}^1_t=
\frac{\widehat X_t(x,0) }{R_t(x,0)}M_t^1(x,0),\quad {\widetilde N}^1_t=\frac{x}{y} N_t^1(y,0),\quad t\in[0,T].
\end{equation}
In addition, the Hessian terms in the quadratic expansion of $u$ and $v$  can be identified as 
\begin{equation}\label{3253}
\begin{array}{rcl}
a(d,d)
&=& \frac{R_0(x,0)}{x}\inf\limits_{\widetilde M\in\mathcal {\widetilde M}^2(x,0)}\mathbb E^{\mathbb {\widetilde R}(x,0)}\left[\left(\widetilde M_T + xF\left({A\left(\widehat X_T(x,0)\right)} - 1 \right)\right)^2 \right]  + 
C_a.\\
&=&  \frac{R_0(x,0)}{x}\mathbb E^{\mathbb {\widetilde R}(x,0)}\left[\left({\widetilde N}^1_T \right)^2 \right] +\frac{R_0(x,0)}{x}P_0^2+ C_a,
\end{array}
\end{equation}
where
\begin{equation}\label{eqC_a}
C_a\triangleq x^2\mathbb E^{\mathbb R(x,0)}\left[F^2\frac{A(\widehat X_T(x,0)) - 1}{A(\widehat X_T(x,0))}- G\right],
\end{equation}
and \begin{equation}
\begin{array}{rcl}
b(d,d)
&=& \frac{R_0(x,0)}{x}\inf\limits_{ {\widetilde N}\in\mathcal{ N}^2(y,0)}\mathbb E^{\mathbb{\widetilde R}(y,0)}\left[ \left({\widetilde N}_T + yF\left( A\left(\widehat X_T(x,0)\right) - 1\right)\right)^2\right] + C_b.\\
&=& \frac{R_0(x,0)}{x}\left (\frac yx  \right)^2 \mathbb E^{\mathbb{\widetilde R}(y,0)}\left[ \left({\widetilde M}^1_T \right)^2\right] +\frac{R_0(x,0)}{x}\left (\frac yx  \right)^2P_0^2+ C_b,
\end{array}
\end{equation}
where
 \begin{equation}\label{eqC_b}
 C_b\triangleq y^2\mathbb E^{\mathbb R(x,0)}\left[G + F^2\left(1 - A\left(\widehat X_T(x,0)\right)\right) \right].
 \end{equation}
The cross terms in the Hessians of $u$ and $v$  are identified as
$$a(x,d)=P_0$$ and $b(y,d)$ is given by 
$$ b(y,d)=\frac y x \frac{P_0}{a(x,x)}.$$
With these identifications, all the  conclusions  of Theorem \ref{mainThm2}  and Corollary \ref{2151} hold true.
\end{thm}
\begin{proof}

Let us prove \eqref{3251} first. 
Completing the square in \eqref{addC}, we get 
\begin{equation}\label{3252}
a(d,d) = \inf\limits_{M\in\mathcal M^2(x,0)}\mathbb E^{\mathbb R(x,0)}\left[A\left(\widehat X_T(x,0)\right)\left(M_T + xF\left(1 - \frac{1}{A(\widehat X_T(x,0))}\right)\right)^2 \right]  + 
C_a,
\end{equation} 
where $C_a$ is defined in \eqref{eqC_a}. As $$\frac{d\mathbb{ R}(x,0)}{{d\mathbb {\widetilde R}(x,0)}} = \frac{A\left(\widehat X_T(x,0)\right)R_0(x,0)}{x} = \frac{\widehat X_T(x,0)R_0(x,0)}{R_T(x,0)x},$$
using Lemma \ref{lem:change-rt},
 we can reformulate \eqref{3252} as
\begin{equation}\label{3253}
\begin{array}{rcl}
a(d,d) &=&  \frac{R_0(x,0)}{x}\inf\limits_{M\in\mathcal M^2(x,0)}\mathbb E^{\mathbb {\widetilde R}(x,0)}\left[\left(M_T\frac{\widehat X_T(x,0)}{R_T(x,0)} + xF\left({A\left(\widehat X_T(x,0)\right)} - 1\right)\right)^2 \right]  + 
C_a,\\
&=& \frac{R_0(x,0)}{x}\inf\limits_{\widetilde M\in\mathcal {\widetilde M}^2(x,0)}\mathbb E^{\mathbb {\widetilde R}(x,0)}\left[\left(\widetilde M_T + xF\left({A\left(\widehat X_T(x,0)\right)} - 1 \right)\right)^2 \right]  + 
C_a.\\
\end{array}
\end{equation}
Likewise,
completing the square in \eqref{bddC}, we obtain
\begin{equation}\label{3254}
\begin{array}{rcl}
b(d,d)&=& \inf\limits_{N\in\mathcal N^2(y,0)}\mathbb E^{\mathbb R(y,0)}\left[ 
B\left(\widehat Y_T(y,0)\right)\left(N_T + yF\frac{1 - B(\widehat Y_T(y,0)}{B(\widehat Y_T(y,0)}\right)^2\right] + C_b,\\
&=& \frac{R_0(x,0)}{x}\inf\limits_{ N\in\mathcal{ N}^2(y,0)}\mathbb E^{\mathbb{\widetilde R}(y,0)}\left[ \left(N_T + yF\left( A\left(\widehat X_T(x,0)\right) - 1\right)\right)^2\right] + C_b,\\
\end{array}
\end{equation}
where $C_b$ is defined in \eqref{eqC_b}.
Now, decomposition \eqref{3251} (where the constant $P_0$ is still to be determined) results from \eqref{3253}, \eqref{3254}, and optimality of $M^1(x,0)$ and $N^1(y,0)$ for \eqref{addC} and \eqref{bddC}, respectively. 
 As $A\left(\widehat X_T(x,0) \right) = \frac{\widehat X_T(x,0)}{R_T(x,0)}$, taking the expectation in \eqref{4141} under $\mathbb {\widetilde R}(x,0)$, we deduce that $P_0 = a(x,d)$.
 Therefore, using \eqref{12154c}, we deduce that $b(y,d)=\frac y x \frac{P_0}{a(x,x)}$.
\end{proof}
\begin{rem}Applying It\^o formula, one can find expressions for the corrections of the optimal proportions in terms of the Kunita-Watanabe decomposition under risk-tolerance wealth process as num\'eraire, in the spirit of Theorem \ref{thmCorOptimizer}. However, in the general case when $\frac{R(x,0)}{R_0(x,0)}=X'(x,0)\neq\widehat X(x,0)/x$ such a correction to proportions also contains the terms $\widehat X(x,0)/R(x,0)$ and $\widetilde{M}^1(x,0)$.
\end{rem}
\begin{rem}

Theorem \ref{riskTolThm2} gives an interpretation of $a(x,d)$ as an  utility-based price. Let us start by observing that 
$$a(x,d) = \mathbb E^{\mathbb{\widetilde R}(x,0)}\left[\left(A(\widehat X_T(x,0)) -1\right) xF\right] = \mathbb E\left[\frac{\left(\widehat X_T(x,0) - R_T(x,0) \right)}{R_0(x,0)} xF\frac{\widehat Y_T(y,0)}{y}\right].$$ 
If there exists a wealth process $X'$ such that 
\begin{equation}\label{3265}
X'_T\geq \left|\left(\widehat X_T(x,0) - R_T(x,0) \right)F\right|,
\end{equation} and
$X' \widehat Y$ is a uniformly integrable martingale\footnote{In particular, such a process $X'$ satisfying both conditions exists if $|F|\leq C$ a.s. for some constant $C>0$. In this case $X' = C(R(x,0) + \widehat X(x,0))$ satisfies \eqref{3265} and $X'\widehat Y(y,0)$ is a $\mathbb P$-martingale by Theorem~\ref{riskTolThmLocal}.}, according to \cite{HK04, HKS05}, $a(x,d)$ represents the marginal utility-based price of
the ``random endowment''  
$\frac{\left(\widehat X_T(x,0) - R_T(x,0) \right)}{R_0(x,0)} xF.$
\end{rem}

%

\section*{An extended remark}\label{secApplication}
Below we will consider an application of our results. As was pointed out in the introduction, there is a number of models, or rather classes of models, which admit a closed form solution, see for example \cite{Liu07} and references therein. Once we perturb the input parameters, the solution typically halts to exist in the closed form. However, Theorems  \ref{mainThm1}, \ref{mainThm2}, \ref{161}, and \ref{thmCorOptimizer} give approximations to the value function, the optimizer, and the optimal trading strategy. We will assume that $U(x) = \frac{x^p}{p}$, $p\in(-\infty, 0)\cup (0,1)$ and there are two traded securities, a money market account with zero interest rate and one traded stock 
that satisfies conditions of \cite{Liu07}. In this case the optimal strategy can be  obtained explicitly, see \cite[Proposition 2]{Liu07}. 
\subsection*{Explicit form of the correction terms}
As we are in power-utility settings, it is enough to consider $x=1$.
We assume that $0$-model admits a solution $\widehat X(1,0)$, where $\widehat X(1,0) = 1 + (\widehat X(1,0)\widehat \pi(1,0))\cdot S^0$, for some predictable and $S^0$-integrable process $\widehat \pi(1,0)$.
Let us set recall that $M^R$ (specified in \eqref{4231}) is given by
$$M^R = M + (\lambda - \widehat \pi(1,0))\cdot \langle M\rangle = S^0 - \widehat \pi(1,0)\cdot\langle S^0\rangle.$$

Let us consider perturbations of $\lambda$ by a process $\nu$, such that  Assumption \ref{integrabilityAssumption} holds. In these settings, the solutions to \eqref{axxC} and \eqref{addC}, respectively, are 
$$M^0(1,0) \equiv 0\quad and \quad M^1(1,0) = \gamma^1\cdot M^R.$$ 
Following the argument in section \ref{approxTradingStrategies}, we specify 
\begin{equation}\nonumber
\begin{array}{rcl}
\tau_\varepsilon &= &\inf\left\{ t\in [0,T]:~|M^1_t(1,0)|\geq \frac{1}{\varepsilon}~or~\langle M^1(1,0)\rangle_t\geq \frac{1}{\varepsilon}\right\},\quad\varepsilon>0, \\
\end{array}
\end{equation} 
and
\begin{equation}\nonumber
\gamma^{1,\varepsilon} = \gamma^11_{\{[0,\tau_\varepsilon]\}},\quad \varepsilon>0.
\end{equation}
In \eqref{4232}, we have
\begin{equation}\nonumber
 X^{\Delta x, \delta, \varepsilon} = (1 + \Delta x)\mathcal E\left(\left(\widehat \pi(1,0) +  \delta(\nu+ \gamma^{1,\varepsilon})\right)\cdot S^\delta\right).
\end{equation} 
Following the argument of Theorem \ref{thmCorOptimizer}, we can find $\varepsilon(\Delta x,\delta)$, such that
\begin{displaymath}
\mathbb E\left[ U\left( X_T^{\Delta x, \delta, \varepsilon(\Delta x,\delta)}\right) \right] = u(1 + \Delta x, \delta) - o(\Delta x^2 + \delta^2).
\end{displaymath} 
Using Theorem \ref{162}, we deduce that the Kunita-Watanabe decomposition of 
$\left(\mathbb E^{\mathbb R(x,0)}\left[\frac{p}{1- p}F|\mathcal F_t \right]\right)_{t\in[0,T]},$ where $F$ is specified in \eqref{defFG}, is:
\begin{equation}\label{291}
\frac{p}{1- p}F = -\frac{a(x,d)}{1-p} + \gamma^1\cdot M^R_T + \frac{1}{(1-p)y}N^1(y,0),
\end{equation}
where $y = u_x(1,0)$ and $N^1(y,0)$ is the solution to \eqref{bddC}.

Moreover, in this case 
the corresponding coefficients $a(x,x)$, $a(x,d)$, and $a(d,d)$ from \eqref{axxC}, \eqref{axdC}, and \eqref{addC}, respectively, are given by
\begin{displaymath}
\begin{array}{lcl}
a(x,x) & = & 1- p,\\
a(x,d) &= & -p\mathbb E^{\mathbb R(1,0)}\left[F\right],\\
a(d,d) &= & \frac{1}{1-p}(a(x,d))^2 + \frac{1}{y^2(1-p)}\mathbb E^{\mathbb R(1,0)}\left[N^1_T(1,0)^2\right] - \mathbb E^{\mathbb R(1,0)}\left[ \frac{p}{1-p} F^2+ G\right].\\
\end{array}
\end{displaymath}

\subsection*{Relation to the risk-tolerance process}
In this case, risk-tolerance wealth process exists for every $x>0$, and is given by $R(x,0) = \frac{\widehat X(x,0)}{1- p} = \frac{x}{1-p}\widehat X(1,0)$, whereas $\mathbb{\widetilde R}(x,0) = \mathbb R(x,0)$, $x>0$, where $\mathbb{\widetilde R}(x,0)$ is defined in \eqref{tildeR}. Theorem \ref{riskTolThmLocal} implies that 
\begin{equation}\nonumber
\begin{array}{rcl}
u_{xx}(x,0) &=& -\frac{y}{x}(1-p),\\
 \lim\limits_{\Delta x\to 0}\frac{1}{\Delta x}\left(
\widehat X_T(x + \Delta x, 0) - \widehat X_T(x,0)\right)&=& \frac{\widehat X_T(x,0)}{x},\\
  \lim\limits_{\Delta y\to 0}\frac{1}{\Delta y}\left(
\widehat Y_T(y + \Delta y, 0) - \widehat Y_T(y,0)\right)&=& \frac{\widehat Y_T(y,0)}{y},\\
\end{array}
\end{equation}
where the convergence take place in $\mathbb P$-probability.
In turn, \eqref{291} also asserts that $M^1_T(x,0)$ and $N^1_T(y,0)$ form (up to multiplicative constants) an orthogonal decomposition of $F + \frac{a(x,d)}{p}$ under $\mathbb R(1,0)$, in accordance with Theorem \ref{riskTolThm2}. 
In particular,
$\mathbb E^{\mathbb R(1,0)}\left[F\right] = \frac{a(x,d)}{-p}.$
\bibliographystyle{alpha}\bibliography{literature1}
\end{document}